\documentclass[11pt,twoside]{article}
\usepackage{tikz}
\usetikzlibrary{positioning}
\usepackage{amsmath,amsthm}
\usepackage{graphicx}
\usepackage{geometry}
\usepackage{multirow}
\usepackage{float}
\usepackage{sectsty}
\usepackage{amssymb}

\usepackage[utf8]{inputenc}
\usepackage[T1]{fontenc}
\usepackage{lmodern}
\usepackage[english]{babel}
\usepackage{braket}
\usepackage{mathrsfs}
\usepackage{stmaryrd }
\usepackage{titleps}
\usepackage{tabularx}
\usepackage{enumitem}
\usepackage{color}

\newcommand{\field}[1]{\mathbb{#1}}
\newcommand{\N}{\field{N}}
\newcommand{\R}{\field{R}}
\newcommand{\C}{\field{C}}

\newcommand{\CS}{\mathcal S}

\newcommand{\HH}{\mathscr H}
\newcommand{\KK}{\mathscr K}

\newcommand{\LL}{\mathscr L}

\newcommand{\eps}{\varepsilon}
\newcommand{\ph}{\varphi}

\newcommand{\ran}{\mathrm{Ran}}

\renewcommand{\d}[1]{\textup{d}#1}

\newcommand{\sgn}{\operatorname{sgn}}

\newcommand{\Rea}{\operatorname{Re}}

\DeclareMathOperator*{\limeps}{\lim\limits_{\eps\to 0}}
\DeclareMathOperator*{\limepsr}{\lim_{\eps\to 0}}

\newtheorem{theorem}{Theorem}[section]
\newtheorem{lemma}[theorem]{Lemma}
\newtheorem{proposition}[theorem]{Proposition}

\newtheorem{kor}[theorem]{Corollary}

\usepackage[babel]{csquotes}

\sectionfont{\Large}
\subsectionfont{\Large}
\date{\today}
\usepackage[normalem]{ulem}

\sectionfont{\Large}
\subsectionfont{\large}

\title{From Short-Range to Contact Interactions in the 1d Bose Gas}

\author{Marcel~Griesemer, Michael~Hofacker and Ulrich~Linden\\  
\small Fachbereich Mathematik, Universit\"at Stuttgart, D-70569 Stuttgart, Germany}  
\date{}

\begin{document}
\maketitle

\setcounter{figure}{0}

\numberwithin{equation}{section}
\newpagestyle{mypage}{%
      \setfoot{}{\thepage}{}
}

\setcounter{page}{1}
\newcommand\Scalefrac[2]{\scalebox{0.85}{$\dfrac{#1}{#2}$}}
\newcommand\sym[0]{\scalebox{0.7}{$\mspace{-6mu}\raisebox{-2mm}{\textup{sym}}$}}
\newcommand\huger[0]{\scalebox{5}{$)$}}
\renewcommand{\abstractname}{Abstract}

\begin{abstract}
For a system of $N$ bosons in one space dimension with two-body $\delta$-interactions the Hamiltonian can be defined in terms of the usual closed semi-bounded quadratic form. We approximate this Hamiltonian in norm resolvent sense by Schr\"odinger operators with rescaled two-body potentials, and we estimate the rate of this convergence. 
\end{abstract} 

\section{Introduction} \label{sec1}
\normalfont
\normalsize

Short range interactions with large scattering length in quantum mechanical systems of bosons or distinguishable particles  are conveniently described by $\delta$-potentials, unless the space dimension is three and the number of particles exceeds two \cite{SolvableModels,AlbKur,BraHam}.  This has a long tradition in physics and rigorous formulation in mathematics \cite{SolvableModels,Figari,DR}. Yet, a mathematical justification of such idealized models based on many-particle Schr\"odinger operators with suitably rescaled two-body potentials is still at the beginning \cite{Basti2018, SeiYin}. In the present paper we address this problem for the system of $N$ bosons in one space dimension. If a trapping potential were included, this would be the Lieb-Liniger model \cite{LiLi}. We show that the Hamiltonian is the limit, in norm resolvent sense, of rescaled Schr\"odinger operators and we estimate the rate of convergence. 

The Hilbert space of the system to be considered is the $N$-fold symmetric tensor product 
\begin{align} \label{DefHilbert}
          \HH:= \otimes^N_{\rm sym} L^2(\R)
\end{align}
and the Hamiltonian is formally given by  
\begin{equation}\label{DefH}
     H= H_0 - \alpha \sum_{i < j}^N \delta(x_j-x_i),
\end{equation}
where $H_0 = -\Delta$ describes the kinetic energy of the bosons, $\alpha \in \R$ determines the interaction strength, and $x_i\in \R$ denotes the position of the $i$th boson. It is well known that $H$ may be self-adjointly realized in terms of a closed semi-bounded quadratic form and that vectors from the domain can be characterized by a jump condition in the first partial derivatives at the collision planes \cite{Basti2018,LiLi}. We are interested in the approximation of $H$ in terms of Schr\"odinger operators of the form
\begin{align}
     H_\eps:= H_0 - g_{\eps} \sum_{i < j}^N V_{\eps}(x_j-x_i), \qquad \quad \eps>0,
\label{DefHeps}
\end{align}
where $V\in L^1\cap L^2(\R)$,  $V(r)=V(-r)$, and
\begin{align}
V_{\eps}(r):=\eps^{-1}\,V(r/\eps), \qquad \quad \eps>0. \label{Vscale}
\end{align}
Of course, the coupling constant $g_{\eps} \in \R$ will be chosen in such a way that 
$$
     g_\eps\int V(r)\,dr \to \alpha \qquad (\eps\to 0).
$$

In the case $N=2$ it is well known and in the case $N=3$ it was recently shown that $H_\eps\to H$ in the norm resolvent sense \cite{SolvableModels,Basti2018}. For general $N\geq 2$ we will see that convergence in the \emph{strong} resolvent sense is easily established with the help of $\Gamma$-convergence. In fact, it is not hard to see that $H_\eps+C\geq 0$ uniformly in $\eps>0$ for some $C>0$. Let $q_\eps$ and $q$ denote the quadratic forms associated with $H_\eps+C$ and $H+C$, respectively. Then $q_\eps\to q$ in the sense of weak and strong $\Gamma$-convergence, and, by an abstract theorem, this is equivalent to convergence $H_\eps\to H$ in the \emph{strong} resolvent sense \cite{DalMaso}. The main result of the present paper is that $H_\eps\to H$ in the \emph{norm} resolvent sense with estimates on the rate of convergence in terms of the decay of $V$. Norm resolvent convergence, unlike the weaker strong resolvent convergence, implies convergence of the spectra \cite{RS1} and convergence of the unitary groups in a (weighted) operator norm (see remark below).

We set 
\begin{align}  \label{HTilde}
      \widetilde{\HH} := L^2_{\text{ev}}(\R,\d{r}) \otimes L^2(\R, \d{R}) \otimes \bigotimes^{N-2}_{\text{sym}} L^2(\R),
\end{align}
where $L^2_{\text{ev}}(\R)\subset L^2(\R)$ denotes the subspace of even functions.
Here, $r$ and $R$ correspond to the relative and center of mass coordinates 
\begin{equation} \label{DefrR}
      r:=x_2-x_1, \qquad R:= \frac{x_1+x_2}{2}
\end{equation}
of the boson positions $x_1$ and $x_2$. We define (possibly unbounded) closed operators  $A_{\eps}: D(A_{\eps}) \subseteq \HH \rightarrow \widetilde{\HH}$ for $\eps >0$ by
\begin{align} \label{DefAeps}
 (A_\eps \Psi)(r,R,x_3,...,x_{N}) := \sqrt{\Scalefrac{(N-1)N}{2}} |V(r)|^{1/2} \Psi(R-\tfrac{\eps r}{2}, R+\tfrac{\eps r}{2}, x_3, ..., x_{N}),
\end{align}
which is nothing but the operator of multiplication by $\sqrt{(N-1)N/2} \left|V_{\eps}(x_2-x_1)\right|^{1/2}$ written in the new coordinates \eqref{DefrR} and followed by the (unitary) rescaling $r\mapsto \eps r$. Let $J$ denote multiplication by $\sgn(V)$ in $L^2_{\text{ev}}(\R,\d{r})$ and let $B_\eps = JA_\eps$. Then  
\begin{align}\label{Hepsnew}
        H_{\eps}=H_0 - g_{\eps} \, A_{\eps}^{*}B_{\eps}, \qquad \quad \eps>0. 
\end{align}
By a general result on self-adjoint operators of this form (Appendix \ref{AppendixB}), if $z\in\rho(H_{\eps})\cap\rho(H_0)$\footnote{Note our definition of the resolvent set at the end of this section.}, then 
\begin{equation}\label{Krein}
      (H_{\eps}+z)^{-1} = (H_0 + z)^{-1} +  (H_0 + z)^{-1}A_{\eps}^{*} \left( g_{\eps}^{-1}- \phi_{\eps}(z)  \right)^{-1}\, B_{\eps} (H_0 + z)^{-1}, 
\end{equation}
where $\phi_{\eps}(z) \in \LL(\widetilde{\HH})$ on $D(A_{\eps}^{*})$ is given by
\begin{align} \label{Defphieps}
       \phi_\eps(z) = B_{\eps}(H_0 + z)^{-1} A_{\eps}^{*}.
\end{align} 
The formula \eqref{Krein} is our starting point for proving resolvent convergence. It allows us to generalize the methods familiar from the case $N=2$ \cite{SolvableModels}.

It is not hard to see that the limit
\begin{align}\label{S(z)} 
     S(z)&= \limeps \; A_{\eps} (H_0+z)^{-1} 
\end{align}
exists for some, and hence for all $z \in \rho(H_0)$. This is independent of the space dimension $d\leq 3$.
The subtle point in two and three space dimensions, even for $N=2$, is the convergence of  $(g_{\eps}^{-1}-\phi_{\eps}(z))^{-1}$, which involves the cancellation of divergencies \cite{SolvableModels,BHS2013}. For $d=3$ and $N\geq 3$ there is, in addition, a partly open problem known as Thomas effect \cite{Thomas1935, MinlosFaddeev1,MinlosFaddeev2,BastiTeta}. In the present paper we avoid these complications by considering $d=1$ only. In this case the limit
\begin{align}\label{Defphi}   
             \phi(z)= \limeps  \phi_{\eps}(z) 
\end{align} 
exists. In combination with \eqref{S(z)}, this allows us to take the limit $\eps\to 0$ in \eqref{Krein} and leads us to the following theorem:

\begin{theorem}\label{theo1}
Let $V \in L^1\cap L^2(\R)$ with $V(-r)=V(r)$, let $H_\eps$ be defined by \eqref{DefHeps} and 
suppose that $g=\lim_{\eps\to 0}g_{\eps}$ exists. 
Then $H_{\eps} \to H$ in the norm resolvent sense as $\eps\to 0$, and 
\begin{equation} \label{Hres}
 (H+z)^{-1} = (H_0 + z)^{-1} + g \, S(\overline{z})^{*} \left(1  - g\phi(z) \right)^{-1} JS(z) 
\end{equation}
for $z \in \rho(H_0)\cap \rho(H)$. If, in addition, $\int |r|^{2s}|V(r)|\, \d{r}<\infty$ and $|g_\eps - g|=O(\eps^s)$ for some $s\in (0,1)$, then $\| (H+z)^{-1}- (H_{\eps}+z)^{-1}\| = O(\eps^s)$ as $\eps\to 0$.
\end{theorem}

\noindent
\emph{Remarks}
\begin{enumerate}
\item The norm convergence established by this theorem implies that 
$$
       \|(e^{-iH_\eps t}- e^{-iHt})(H+i)^{-1}\| \to 0\qquad (\eps\to 0),
$$ 
uniformly on compact (or growing, if $s>0$) time intervals. In contrast, strong resolvent convergence implies a similar result in the strong operator topology \cite{Pazy}.

\item The operators $\phi(z)$ and $S(z)$ depend on $V$ and so the left hand side of \eqref{Hres} seems to depend on $V$ as well. This apparent dependence may be removed by integrating out the potential in the second term of \eqref{Hres}, see Section~\ref{sec5}. In particular, this term vanishes if $\int V(r)\,\d{r}=0$.
\end{enumerate} 

Our proofs of \eqref{S(z)} and \eqref{Defphi}, and hence of Theorem~\ref{theo1}, rely on explicit expressions for the integral kernels of $A_{\eps}(H_0+z)^{-1}$ and $\phi_{\eps}(z)$ in terms of the Green's function $G^n_z$ of $-\Delta + z$ in $\R^n$.
This procedure is fairly involved in the case of $\phi_{\eps}(z)=\sum_{i<j}\phi_{\eps}^{ij}(z)$ because the kernel of $\phi_{\eps}^{ij}(z)$ depends on the pair $(i,j)$ of particles. The bosonic symmetry is lost, in part, because of the symmetry breaking choice  \eqref{DefrR} of coordinates. Once we have shown convergence of the resolvent $(H_\eps+z)^{-1}$, to conclude the proof of the first statement of the theorem it suffices to show that $H_\eps\to H$ in the strong resolvent sense. By a general theorem \cite{DalMaso}, this is equivalent to strong and weak $\Gamma$-convergence of the associated quadratic forms $q_\eps$ and $q$, which we prove in Appendix~\ref{sec:Gamma}.

The main elements of our approach, such as the representation \eqref{Hepsnew} and the Krein formula \eqref{Krein} are independent of the space dimension and the statistics of the particles.  A result similar to Theorem~\ref{theo1} for (distinguishable) particles with short-range interactions in two dimensions is in preparation. This is related to, yet distinct from work described in \cite{Figari,DR,GriesemerLinden2}, where two-dimensional systems with contact interactions are approximated by systems with ultraviolet cutoff.

A result similar to Theorem~\ref{theo1} for three distinct particles in one dimension was previously established in  \cite{Basti2018}. The proof in  \cite{Basti2018}, however, relies on Fadeev equations, which do not generalize to $N>3$. In another closely related work, the Lieb-Liniger model with repulsive $\delta$-interactions is derived from a trapped $3d$ Bose gas with non-negative two-body potentials \cite{SeiYin}.

The proof of Theorem \ref{theo1} is given in Section~\ref{sec5}. The Sections \ref{sec2}, \ref{sec3} and \ref{sec4} provide all preparations apart from generalities, which we collect in the appendix. In Appendix \ref{AppendixA} we collect the basic properties of the Green's function $G^n_z$ along with some nonstandard inequalities. In Appendix~\ref{AppendixB} the Krein formula \eqref{Krein}  is established in an abstract framework, and in Appendix \ref{sec:Gamma} we prove the $\Gamma$-convergence $q_\eps\to q$.

\medskip\noindent
\emph{Notations.} In this paper the resolvent set $\rho(H)$ of a closed operator $H$ is defined as the set of $z\in \C$ for which $H+z :D(H)\subset\HH\to\HH$ is a bijection. This differs by a minus sign from the conventional definition. The $L^2$-norm will be denoted by $\|\cdot \|$, without index, while all other norms carry the space as an index, as  e.g. in $\|V\|_{L^1}$.


\section{Auxiliary operators}  \label{sec2}

This section defines auxiliary operators that will be helpful in the proofs of \eqref{S(z)} and \eqref{Defphi}. 

The change of coordinates \eqref{DefrR} is implemented by the
coordinate transformation $\KK: \HH \rightarrow \widetilde{\HH}$ defined by
\begin{align}
(\KK \Psi)(r,R,x_3,...,x_{N}):= \Psi\left(R-\tfrac{r}{2},R+\tfrac{r}{2}, x_3,...,x_{N} \right). \label{DefK} 
\end{align}
It follows that  $\KK^{*}: \widetilde{\HH} \rightarrow \HH$ is given by
\begin{align}
(\KK^{*} \widetilde{\Psi})(x_1,...,x_{N})= \frac{2}{N(N-1)} 
 \sum_{i < j} \widetilde{\Psi}\left(x_j-x_i,\Scalefrac{x_i+x_j}{2},x_{1},...\hat{x}_i...\hat{x}_j...,x_N \right), \label{DefK*}
\end{align}
where the hat in $\hat{x}_i$ indicates omission of this variable. Since terms arising from distinct pairs 
$(i,j)$ in \eqref{DefK*} will be treated separately later on, we further introduce for $1\leq i<j\leq N$ the operator $\KK_{ij}^{*}: \widetilde{\HH} \rightarrow L^2(\R^{N})$ by 
\begin{align}
(\KK_{ij}^{*} \widetilde{\Psi})(x_1,...,x_{N}):= \widetilde{\Psi}\left(x_j-x_i,\Scalefrac{x_i+x_j}{2},x_{1},...\hat{x}_i...\hat{x}_j...,x_N \right) \mspace{-5mu}. \label{DefKij}
\end{align}
The asterisk in $\KK_{ij}^{*}$ is part of the notation, which reminds us of the decomposition
\begin{align} \label{KKdecomp}
       \KK^{*}  = \frac{2}{N(N-1)}  \sum_{i < j}\KK_{ij}^{*}. 
\end{align}
It does not have the meaning of adjoint.

Let now $V \in L^1\cap L^2(\R)$ be a given even potential, let $v=|V|^{1/2}$ and let $u=\sgn(V)v$ so that $V=vu$. Let 
$U_{\eps}$ denote the unitary scaling in $ \widetilde{\HH}$ defined by
\begin{equation} \label{DefUeps}
     \left(U_{\eps} \widetilde{\Psi}\right)(r,R,x_3,...,x_N):= \eps^{1/2}\,\widetilde{\Psi}(\eps r, R,x_3,...,x_N).
\end{equation}
The closed operator
$A_{\eps}: D(A_{\eps}) \subseteq \HH \rightarrow \widetilde{\HH}$ defined by
\begin{align}\label{Aepsnew}
        A_{\eps}:= \sqrt{\frac{(N-1)N}{2}} (v \otimes 1) \,\eps^{-1/2} U_{\eps} \KK 
\end{align}
agrees with \eqref{DefAeps}. With the help of  \eqref{DefK*}, it is straightforward to
verify that
\begin{align} \label{Potdecomp}
    \sum_{i < j}^N V_{\eps}(x_j-x_i) = \frac{(N-1)N}{2} \KK^{*}(V_{\eps} \otimes 1) \KK = A_{\eps}^{*}B_{\eps}
\end{align}
on $D(H_0)$, which proves \eqref{Hepsnew}. It follows, in particular, that $A_{\eps}^{*}B_{\eps}$ and $A_{\eps}^{*}A_{\eps}$ are infinitesimally $H_0$-bounded. Hence Theorem~\ref{thm:Krein} applies to \eqref{Hepsnew}, which justifies \eqref{Krein}.

\section{The limit of $A_{\eps}(H_0+z)^{-1}$}  
\label{sec3}

In this section the limit of $A_{\eps}(H_0+z)^{-1}: \HH \rightarrow \widetilde{\HH}$  as $\eps\to 0$ is computed assuming $V \in L^1(\R)$ only.  The rate of convergence is estimated in terms of the decay of $V$ at $r=\infty$. While $A_{\eps}$ will be an unbounded operator in general, the operator $A_{\eps}(H_0+z)^{-1}$ is bounded as will be seen. In this section the restriction to $d=1$ space dimension would not be necessary, all arguments go through for general $d\leq 3$.

The Laplacian $H_0$ expressed in the relative and center of mass coordinates \eqref{DefrR} reads
\begin{align}
       \widetilde{H}_0 = -2\Delta_r - \frac{\Delta_R}{2} + \sum\limits_{i=3}^{N} -\Delta_{x_i}. \label{DefH0Tilde}
\end{align} 
In terms of the coordinate transformation $\KK$ from \eqref{DefK} this means that
\begin{align}
       \KK (H_0+z)^{-1} = (\widetilde{H}_0+z)^{-1}\KK. \label{H0H0Tilde}
\end{align}
Hence, \eqref{Aepsnew} implies that
\begin{align}
       A_{\eps}(H_0+z)^{-1} = \sqrt{\frac{(N-1)N}{2}} \,T_{\eps}(z)\KK \label{AepsTeps}
\end{align}
with an operator $T_{\eps}(z)$ in $\widetilde{\HH}$ defined by
\begin{align}\label{DefTeps}
      T_{\eps}(z):=\left( v \otimes 1 \right) \eps^{-1/2}U_{\eps} (\widetilde{H}_0 + z)^{-1}. 
\end{align}
It remains to prove existence of the limit $\lim_{\eps\to 0}T_{\eps}(z)$.

Upon a Fourier transform in $(R,x_3,...,x_N)$, the operator \eqref{DefTeps} acts pointwise in the associated momentum variable $\underline{P}=(P,P_3,...,P_N)$ by an operator $T_{\eps}(z,\underline{P})$ that is given by
\begin{align}
      T_{\eps}(z,\underline{P})=&  \,\frac{1}{2} \left( v \otimes 1 \right) \eps^{-1/2}U_{\eps} \left(-\Delta_r + \frac{z +Q}{2}\right)^{-1}, \label{TepsP} \\
        Q:=&\frac{P^2}{2}+ \displaystyle\sum\limits_{i=3}^{N} P_i^2. \label{DefQ}
\end{align}
By \eqref{TepsP}, the integral kernel associated with $T_{\eps}(z,\underline{P})$ is
\begin{align}\label{Tepskern}
         \Scalefrac{1}{2} \,v(r) \; G_{\frac{1}{2}(z+Q)}\mspace{-2mu}\left(\eps r -r'\right),  
\end{align}
where $G_{\lambda}:=G_{\lambda}^1$ denotes the Green's function of $-\Delta + \lambda: H^2(\R) \rightarrow L^2(\R)$, which is explicitly given by
\begin{align}
        G_{\lambda}(x) = \frac{\exp(-\sqrt{\lambda}\left| x  \right|) }{2 \sqrt{\lambda}}, \qquad \lambda>0. \label{G^1_z}
\end{align}
Since $G_{\lambda} \in L^2(\R)$, the assumption that $V \in L^1( \R)$ (and thus $v \in L^2(\R)$) implies that $T_{\eps}(z,\underline{P})$ is a Hilbert-Schmidt operator. Let $T_0(z,\underline{P})$, and thus $T_0(z)$, be defined by \eqref{Tepskern} with $\eps=0$.
We expect that $T_{\eps}(z)$ converges to $T_0(z)$ as $\eps\to 0$. The following two lemmas are concerned with this convergence. The first step is to show that it suffices to consider
potentials $V$ with compact support. For this purpose, we introduce for any $k>0$ the cutoff potential
\begin{equation}
V_k(x):= \begin{cases}
V(x) &\textup{if}\;\, |x| \leq k \\
0    &\textup{otherwise}
\end{cases} \label{cutpot}
\end{equation}
and we set $v_k(r):=\left|V_k(r)\right|^{1/2}$. By $T_{\eps,k}(z)$ and $T_{0,k}(z)$  we denote the operators $T_{\eps}(z)$ and $T_{0}(z)$ with $v$ replaced by $v_k$, respectively.
The corresponding kernels $T_{\eps,k}(z,\underline{P})$ and $T_{0,k}(z,\underline{P})$ are given by \eqref{Tepskern} with $v_k$ instead of $v$. The next lemma shows that $T_{\eps,k}(z)$ is close to $T_{\eps}(z)$ uniformly in $\eps \geq 0$ for large $k>0$.

\begin{lemma}\label{lm1}
Let $k\geq 0$, $\eps\geq0$ and $z >0$. Then 
\begin{align}
\| T_{\eps}(z) - T_{\eps,k}(z) \| &\leq \Scalefrac{1}{2} \,
\| V-V_k \|_{L^1}^{1/2} \|G_{z/2} \|. \label{Abs1} 
\end{align}
\end{lemma}

\begin{proof}
Let $k\geq 0$ be fixed and note that $v-v_k=\left|V-V_k\right|^{1/2}$. Hence, it follows from \eqref{Tepskern} that the kernel of $T_{\eps}(z,\underline{P}) - T_{\eps,k}(z, \underline{P})$ reads
\begin{align}
\Scalefrac{1}{2} \left|V(r)-V_k(r)\right|^{1/2} G_{\frac{1}{2}(z+Q)}\left(\eps r -r'\right) \label{Tdiffkern}
\end{align}
with $Q \geq 0$ defined by \eqref{DefQ}. The $L^2$-norm thereof is
\begin{align}
\Scalefrac{1}{2}\,\| V-V_k \|_{L^1}^{1/2} \big\|G_{\frac{1}{2}(z+Q)} \big\|, 
\end{align}
where $\big\| G_{\frac{1}{2}(z+Q)} \big\| \leq \| G_{z/2}\|$ by \eqref{G^1_z}. 
Hence, the Hilbert-Schmidt norm of $T_{\eps}(z,\underline{P}) - T_{\eps,k}(z, \underline{P})$
is bounded by the right side of \eqref{Abs1}, which is independent of $\underline{P}$. This proves the lemma.
\end{proof}

\noindent
\emph{Remark.} From \eqref{Abs1} with $k=0$ or directly from  \eqref{Tepskern} it follows that $ T_{\eps}(z)$ is bounded with
\begin{align}
       \| T_{\eps}(z) \| &\leq \Scalefrac{1}{2} \, \| V\|_{L^1}^{1/2} \|G_{z/2} \|. \label{Tepsnorm} 
\end{align}

\begin{proposition}\label{lm2}
          Let $z>0$. Then $T_{\eps}(z)$ converges in operator norm to  $T_0(z)$ as $\eps\to 0$. If the condition $\int \d{r}\, |r|^{2s} \left|V(r)\right| < \infty$ is satisfied for some $s\in(0,1]$, then $\|T_{\eps}(z) - T_{0}(z)  \| = O(\eps^s)$ as $\eps\to 0$.  
\end{proposition}

\begin{proof}
Let us first assume that $\int \d{r}\, |r|^{2s} \left|V(r)\right| < \infty$ for some $s\in (0,1]$. 
For fixed $\underline{P}$, it follows from \eqref{Tepskern} that the difference $T_{\eps}(z,\underline{P}) - T_{0}(z, \underline{P})$ is associated with the kernel
\begin{align}
\Scalefrac{1}{2}\, v(r) \left( G_{\frac{1}{2}(z+Q)}\mspace{-2mu}\left(\eps r -r'\right) - G_{\frac{1}{2}(z+Q)}\mspace{-2mu}\left(r'\right)  \right),
\end{align}
where $Q \geq 0$ is defined by \eqref{DefQ}. Using Lemma \ref{lmA2}, the $L^2$-norm thereof can be bounded by
\begin{align}
&\frac{\eps^s }{2} \left(\int \d{r}\, |r|^{2s} \left|V(r)\right|\right)^{1/2} \sup_{r \neq 0}\left((\eps |r|)^{-s}  \Big\|G_{\frac{1}{2}(z+Q)}(\eps r - \cdot) -   G_{\frac{1}{2}(z+Q)} \Big\|\right) \nonumber \\
& \qquad \leq \frac{\eps^s }{2} \left(\int \d{r}\, |r|^{2s} \left|V(r)\right|\right)^{1/2} \sup_{x \neq 0}\left( |x|^{-s} \Big\|G_{z/2}(\cdot + x) -   G_{z/2} \Big\| \right),  
\end{align}
where the right side is independent of $\underline{P}$.
To prove that  $\|T_{\eps}(z)-T_{0}(z) \|=O(\eps^s)$ as $\eps\to 0$, it is thus sufficient to show that for every  $\lambda>0$ there exists a constant $C(\lambda)>0$ such that
\begin{align}
 \forall x \in \R: \qquad \Big\|G_{\lambda}(\cdot + x) -   G_{\lambda} \Big\| \leq  C(\lambda) \min(1,|x|). \label{GreenL2diff}
\end{align}
Using that $G_{\lambda} \in L^2(\R)$ and $\widehat{G_{\lambda}}(p)=(2\pi)^{-1/2}(p^2+\lambda)^{-1}$, a Fourier transform together with the elementary inequality $|\exp(ipx)-1|\leq \min(2,|p||x|)$ yields that
\begin{align}
\Big\|G_{\lambda}(\cdot + x) -   G_{\lambda} \Big\| = (2\pi)^{-1/2} \left( \int \d{p}\, \frac{|\exp(ipx)-1|^2}{(p^2+\lambda)^2} \right)^{1/2} \leq C(\lambda) \min(1,|x|),  \label{GreenL2est}
\end{align}
where
\begin{align}
C(\lambda):= (2\pi)^{-1/2} \max\left(\mspace{-1mu} 4 \mspace{-2mu} \int \mspace{-3mu} \d{p}\, (p^2+\lambda)^{-2}, \int \mspace{-3mu} \d{p}\, \frac{p^2}{(p^2+\lambda)^2}\right)^{1/2} \mspace{-3mu} =\max\left( \lambda^{-3/4}\mspace{-2mu},(16 \lambda)^{-1/4}\right). \label{C(lambda)}
\end{align}
This proves \eqref{GreenL2diff} and hence the second part of the lemma.

In the case of general $V \in L^1\cap L^2(\R)$, we use
 an approximation argument together with Lemma \ref{lm1}. This reduces the proof to showing that
\begin{align}
\limeps \;T_{\eps,k}(z) =  T_{0,k}(z) \label{ziel1}
\end{align}  
for every $k>0$. But this is clear from the above, because $\int \d{r}\, |r|^{2s} \left|V_k(r)\right| < \infty$.
\end{proof} 

\noindent
\emph{Remark.} Lemma \ref{lm2} implies that the limit
\begin{align}
S(z)&= \limeps \; A_{\eps} (H_0+z)^{-1} =\sqrt{\frac{(N-1)N}{2}} \,T_{0}(z) \KK \label{LimitS(z)} 
\end{align} 
exists for every $z >0$.

\section{Convergence of $\phi_{\eps}(z)$} 
\label{sec4}

In this section, as in the previous one, the assumption $V\in L^1(\R)$ will be sufficient. It ensures that $A_{\eps}$ is a densely defined, closed operator from $\HH$ to $\widetilde{\HH}$ and that $A_{\eps}(H_0+z)^{-1}$ is bounded. In the following $\CS(\R^{N})$  denotes the Schwartz space and all operators are introduced on the subspace $\CS(\R^{N}) \cap \widetilde{\HH} \subseteq D(A_{\eps}^{*})$, which is dense in $\widetilde{\HH}$. We will see, however, that some of them have bounded extensions.

Let $z>0$ be fixed. In view of the identities \eqref{Defphieps}, \eqref{KKdecomp}, and \eqref{Aepsnew}, we find the decomposition
\begin{align}
\phi_{\eps}(z)&= B_{\eps}(H_0+z)^{-1}A_{\eps}^{*} \nonumber \\
              &= \frac{N(N-1)}{2}\,\eps^{-1} (u \otimes 1) U_{\eps}\KK(H_0+z)^{-1} \KK^{*} U_{\eps}^{*} (v \otimes 1)\nonumber \\
              &=  \sum_{1 \leq i < j\leq N} \eps^{-1} (u \otimes 1) U_{\eps}\KK(H_0+z)^{-1} \KK_{ij}^{*}U_{\eps}^{*} (v \otimes 1) \nonumber \\
      &=  \sum_{1 \leq i < j\leq N} \phi^{ij}_{\eps}(z), \label{phiepsdecomp}
\end{align}
which defines the operators $\phi^{ij}_{\eps}(z): \CS(\R^{N}) \cap \widetilde{\HH} \rightarrow L^2(\R^{N})$ for $1 \leq i<j \leq N$.
For the further analysis of these operators, we fix a pair $(i,j)$ with $1 \leq i<j \leq N$ and we compute the kernel of $\phi^{ij}_{\eps}(z)$ in terms of the
Green's function $G_z^N$ of $-\Delta + z: H^2(\R^N) \rightarrow L^2(\R^N)$. 
Inserting the defining relations \eqref{DefUeps}, \eqref{DefK},  and \eqref{DefKij} for $U_{\eps}, \KK$ and $\KK_{ij}^{*}$, respectively, we obtain that
\begin{align}
& \left( \phi^{ij}_{\eps}(z)\Phi\right)(r,R,x_3,...,x_{N}) 
 \nonumber \\ &\qquad= \eps^{-1} \, u(r)
 \int \mspace{-3mu} \textup{d}x_1' \cdots \textup{d}x_N' \;
 G_{z}^{N} \mspace{-2mu} \left(R-\tfrac{\eps r}{2}-x_1', R+\tfrac{\eps r}{2}-x_2', x_3-x_3', ..., x_{N}-x_N'\right) \nonumber \\
 & \qquad \;\;\; \times v\left(\Scalefrac{x_j'-x_i'}{\eps}\right) \, \Phi\left(\Scalefrac{x_j'-x_i'}{\eps}, \Scalefrac{x_i'+x_j'}{2}, x_1',...\widehat{x_i'}...\widehat{x_j'}...,x_N' \right) \nonumber \\ &\qquad=  u(r) \int \mspace{-3mu} \textup{d}x_1' \cdots \textup{d}x_N' \,\d{r'} \, \d{R'}  \;  G_{z}^{N} \mspace{-2mu} \left(R-\tfrac{\eps r}{2}-x_1', R+\tfrac{\eps r}{2}-x_2', x_3-x_3', ..., x_{N}-x_N'\right) \nonumber \\
 & \qquad \;\;\; \times v\left(r'\right) \, \Phi\left(r', R', x_1',...\widehat{x_i'}...\widehat{x_j'}...,x_N' \right) \, \delta(x_i'-R'+\tfrac{\eps r'}{2}) \; \delta(x_j'-R'-\tfrac{\eps r'}{2}). \label{Defphiepsij}
\end{align}
Here, the second equation results from the substitution
\begin{align}
r':=\Scalefrac{x_j'-x_i'}{\eps}, \qquad R':=\Scalefrac{x_i'+x_j'}{2}, \nonumber
\end{align}
where two more integrations, which are compensated by the two $\delta$-distributions, were introduced. 
A priori the operators $\phi^{ij}_{\eps}(z)\; (1 \leq i<j \leq N)$ are all unbounded, but we will see in the subsequent lemmata that they all extend to bounded operators given by the same integral kernels. Hence, we expect that they converge to the (formal) limit operators  
$\phi^{ij}_{0}(z)$, which are defined in terms of the corresponding kernels with $\eps=0$. In the following, we prove that this is indeed the case. For this purpose, we divide these operators into the three groups $(i,j) = (1,2)$, $i \in \{1,2\}$ and $j\geq 3$, and $\{i,j\} \subset \{3,\ldots,N\}$, and we analyse these groups separately.

\subsection{The limit of $\phi_{\eps}^{12}(z)$} 
\label{sec41}

For the kernel of $\phi^{12}_{\eps}(z)$ we shall not use \eqref{Defphiepsij} but instead we derive a simpler expression as follows. By the defining expression in \eqref{phiepsdecomp} and  by \eqref{H0H0Tilde},
\begin{align}
\phi^{12}_{\eps}(z) &= \eps^{-1} (u \otimes 1) U_{\eps}(\widetilde{H}_0 + z)^{-1}U_{\eps}^{*} (v \otimes 1) \label{phieps12}
\end{align} 
because $\KK\KK_{12}^{*}=1$ in $L^2(\R^N)$. It follows from \eqref{phieps12}, after a Fourier transform in $(R,x_3,...,x_N)$, that the operator $\phi_{\eps}^{12}(z)$ acts pointwise in the associated momentum variable $\underline{P}=(P,P_3,...,P_N)$ by the operator
\begin{align}
\phi_{\eps}^{12}(z,\underline{P}):= (2\eps)^{-1}   \left( u \otimes 1 \right) U_{\eps}\left(-\Delta_r + \frac{z + Q}{2}\right)^{-1} \mspace{-5mu} U_{\eps}^{*} \left( v \otimes 1 \right), \label{phieps12P} 
\end{align}
where $Q \geq 0$ is defined by \eqref{DefQ}. This operator has the integral kernel
\begin{align}
\Scalefrac{1}{2} \,u(r) \; G_{\frac{1}{2}(z+Q)}\mspace{-2mu}\left(\eps(r-r')\right) v(r'),  \label{phieps12kern}
\end{align}
where $G_\lambda=G_{\lambda}^1$ denotes the one-dimensional Green's function, which is explicitly given by \eqref{G^1_z}. Due to the facts that $u,v \in L^2(\R)$ and $G_{\frac{1}{2}(z+Q)}$ is bounded, we see that $\phi_{\eps}^{12}(z,\underline{P})$ is a Hilbert-Schmidt operator and we expect, and prove below, that $\limepsr \phi_{\eps}^{12}(z)=\phi_{0}^{12}(z)$, where $\phi_{0}^{12}(z,\underline{P})$ is defined in terms of the kernel \eqref{phieps12kern} with $\eps=0$, which is
\begin{align}
u(r) \; \frac{1}{2 \sqrt{2(z+Q)}} \; v(r').\label{phi12kern}
\end{align}
As in the previous section, the first step in the analysis of the limit $\eps\to 0$ is to reduce the problem to the case of compactly supported potentials. Let $V_k\;(k>0)$
denote the cutoff potential introduced in \eqref{cutpot} and let $\phi_{\eps,k}^{12}(z)$ and $\phi_{0,k}^{12}(z)$ denote the operators $\phi_{\eps}^{12}(z)$ and $\phi_{0}^{12}(z)$ with $V$ replaced by $V_k$, respectively. The corresponding kernels are given by \eqref{phieps12kern} and \eqref{phi12kern}, respectively, by substituting $u \rightarrow u_k$ and $v \rightarrow v_k$, where  $v_k(r)=\left|V_k(r)\right|^{1/2}$ and $u_k(r)=\sgn(V_k(r))v_k(r)$. 
The next lemma shows that $\phi^{12}_{\eps}(z)$ and $\phi^{12}_0(z)$ define bounded operators and $\phi^{12}_{\eps,k}(z)$ is close to $\phi^{12}_{\eps}(z)$ uniformly in $\eps\geq 0$ for large $k>0$.

\begin{lemma}\label{lm3}
Let $z>0$ and $\eps \geq 0$. Then $\phi^{12}_{\eps}(z)$ extends to a bounded operator from $\widetilde{\HH}$ to $L^2(\R^{N})$ satisfying the norm estimate
\begin{align}
\| \phi^{12}_{\eps}(z) \| &\leq (2\sqrt{2z})^{-1} \| V \|_{L^1}. \label{Abs2} 
\end{align}
Furthermore, for any $k>0$, we have the estimate
\begin{align}
\| \phi^{12}_{\eps}(z) -\phi^{12}_{\eps,k}(z) \| &\leq  (2\sqrt{z})^{-1}
\| V \|_{L^1}^{1/2} \| V - V_k \|_{L^1}^{1/2}. \label{Abs3} 
\end{align}
\end{lemma}

\begin{proof}
For fixed $z>0$, $\eps \geq 0$ and $\underline{P}$, it follows from \eqref{phieps12kern} and the $L^{\infty}$-bound  
\begin{align}
\big\| G_{\frac{1}{2}(z+Q)} \big\|_{L^{\infty}} \leq (\sqrt{2z})^{-1}  \label{supnormG_z}
\end{align}
that the Hilbert-Schmidt norm of  $\phi^{12}_{\eps}(z,\underline{P})$ is bounded by the right side of \eqref{Abs2}, which is independent of $\underline{P}$. This proves the first part of the lemma.

For the second part of the lemma, we note that $\phi^{12}_{\eps}(z,\underline{P}) -\phi^{12}_{\eps,k}(z,\underline{P})$ has the kernel
\begin{align}
\Scalefrac{1}{2}\left( u(r)v(r') - u_k(r)v_k(r')  \right)G_{\frac{1}{2}(z+Q)}\mspace{-2mu}\left( \eps (r-r')\right). \label{phi12diffkern}
\end{align}
With the help of \eqref{supnormG_z} and the relation 
\begin{align}
\left(  u(r)v(r') -u_k(r)v_k(r')\right)^2= |V(r)V(r')| - |V_k(r)V_k(r')|, \label{Videntity}
\end{align}
we see that the $L^2$-norm of \eqref{phi12diffkern} can be bounded by
\begin{align}
(2\sqrt{2z})^{-1} \left( \int \d{r} \, \d{r'} \, |V(r)V(r')| - |V_k(r)V_k(r')| \right)^{1/2} &= (2\sqrt{2z})^{-1} \left( \|V\|_{L^1}^2 - \|V_k\|_{L^1}^2 \right)^{1/2}
\nonumber \\ &\leq (2\sqrt{z})^{-1} \| V \|_{L^1}^{1/2} \| V - V_k \|_{L^1}^{1/2}. \nonumber
\end{align}  
This shows that the Hilbert-Schmidt norm of $\phi^{12}_{\eps}(z,\underline{P}) -\phi^{12}_{\eps,k}(z,\underline{P})$ is bounded by the right side of \eqref{Abs3}, which is independent of $\underline{P}$.
Hence, \eqref{Abs3} is established.
\end{proof}

Now, we can prove that $\phi^{12}_0(z)$, which is defined by the kernel \eqref{phi12kern}, is the limit of $\phi^{12}_{\eps}(z)$:

\begin{proposition}\label{lm4}
Let $z>0$. Then $\phi^{12}_{\eps}(z)$ converges in operator norm to $\phi^{12}_0(z)$ as $\eps\to 0$. If
$\int \d{r}\, |r|^{2s} \left|V(r)\right| < \infty$ for some $s\in (0,1]$, then $\|\phi^{12}_{\eps}(z) - \phi^{12}_0(z)  \| = O(\eps^s)$ as $\eps\to 0$.  
\end{proposition} 

\begin{proof}
Let us first assume that $\int \d{r}\, |r|^{2s} \left|V(r)\right| < \infty$ for some $s\in (0,1]$. 
Then we note that the kernel of $\phi^{12}_{\eps}(z,\underline{P}) -  \phi^{12}_0(z,\underline{P})$ is for fixed $z>0$ and $\underline{P} \in \R^{N-1}$  given by
\begin{align}
\Scalefrac{1}{2}\, u(r) \left( G_{\frac{1}{2}(z+Q)}\mspace{-2mu}\left(\eps(r-r')\right) - G_{\frac{1}{2}(z+Q)}\mspace{-2mu}\left(0\right)  \right)v(r').
\end{align}
With the help of the elementary inequality $1-\exp(-x)\leq x^s$, which is valid for $x\geq 0$, and the explicit formula \eqref{G^1_z} for $G_z$, we find that
\begin{align}
\left|  G_{\frac{1}{2}(z+Q)}\mspace{-2mu}\left(\eps (r-r')\right) - G_{\frac{1}{2}(z+Q)}\mspace{-2mu}\left(0\right) \right|^2 
&\leq 2^{-s-1}(z+Q)^{s-1}\eps^{2s}|r-r'|^{2s}\nonumber \\
&\leq 2^{-1}z^{s-1} \eps^{2s} (|r|^{2s}+ |r'|^{2s}). 
\end{align}
Using this to estimate the Hilbert-Schmidt norm of $\phi^{12}_{\eps}(z,\underline{P}) -  \phi^{12}_0(z,\underline{P})$, we find that
\begin{align}
\| \phi^{12}_{\eps}(z,\underline{P}) -  \phi^{12}_0(z,\underline{P}) \|^2 \leq 
\frac{\eps^{2s}}{4} z^{s-1} \|V\|_{L^1} \int \mspace{-3mu} \d{r}\, |r|^{2s} \left|V(r)\right| ,
\end{align}
which proves the second part of the lemma.
 
In the case of general $V \in L^1\cap L^2(\R)$, by Lemma \ref{lm3}, it suffices to prove that
\begin{align}
\limeps \;\phi^{12}_{\eps,k}(z) =  \phi^{12}_{0,k}(z) \label{ziel2}
\end{align}  
for every fixed $k>0$. This is clear from the above because $\int |r|^{2s} \left|V_k(r)\right|\,\d{r} < \infty$.
\end{proof}

\subsection{The limits of $\phi^{1j}_{\eps}(z)$ and $\phi^{2j}_{\eps}(z)$ for $j\geq 3$} 
\label{sec42}

Next we discuss the operators $\phi^{1j}_{\eps}(z)$ and $\phi^{2j}_{\eps}(z)$ with $j \in \{3,...,N\}$. The associated kernels will be deduced from \eqref{Defphiepsij}. After the evaluation of the $\delta$-distributions in $x_1'$ and $x_j'$ followed by the substitution $x_2' \rightarrow x_j'$, we infer that the kernel of $\phi^{1j}_{\eps}(z)$ reads
\begin{align}
 u(r)\,  G_{z}^{N} \mspace{-3mu} \left(X_{\eps}^{1j},x_3-x_3',...\widehat{x_{j}-x_{j}'}...,x_{N}-x_N'\right)\mspace{-2mu}v\mspace{-2mu}\left(r'\right),  \label{phieps1jkern}
\end{align}
where
\begin{align}
X^{1j}_{\eps}:=\left(R-R'-\tfrac{\eps (r-r')}{2}, R+\tfrac{\eps r}{2}-x_j',x_j-R'-\tfrac{\eps r'}{2}\right) \label{DefX^1j}
\end{align}
for short. Hence, $\phi^{1j}_{\eps}(z)$ simply acts by convolution in the variables $(x_3,...\widehat{x_j}...,x_N)$.
Consequently, it follows from Lemma \ref{lmA1} (vi) that $\phi^{1j}_{\eps}(z)$ acts pointwise in the associated momentum variables $\underline{P}_j=(P_3,...\widehat{P_j}...,P_N)$ by an operator $\phi^{1j}_{\eps}(z, \underline{P}_j)$ with kernel
\begin{align}
u(r)\; G^{3}_{z+Q_j}\mspace{-3mu}\left(X^{1j}_{\eps}\right)v(r'), \qquad  \quad Q_j:=\displaystyle\sum\limits_{\substack{i=3 \\ i \neq j}}^{N} P_i^2. \label{phieps1jkernP}
\end{align}
This kernel depends on the three-dimensional Green's function, which is explicitly given by  
\begin{align}
G^3_z(x)= \frac{\exp(-\sqrt{z} \left| x \right|)}{4\pi \left| x \right|},\qquad  \quad 0 \neq x \in \R^3, z>0. \label{G3}
\end{align}

Similarly, the operator $\phi^{2j}_{\eps}(z)$ acts pointwise in $\underline{P}_j$ by the operator $\phi^{2j}_{\eps}(z, \underline{P}_j)$ with kernel
\begin{align}
u(r)\; G^{3}_{z+Q_j}\mspace{-3mu}\left(X_{\eps}^{2j}\right)v(r'), \label{phieps2jkernP}
\end{align}
where
\begin{align}
X^{2j}_{\eps}:=\left(R-R'+\tfrac{\eps (r+r')}{2}, R-\tfrac{\eps r}{2}-x_j',x_j-R'-\tfrac{\eps r'}{2}\right).\label{DefX^2j}
\end{align}
A comparison of \eqref{phieps1jkernP} and \eqref{phieps2jkernP} shows that the kernels of the operators $\phi^{1j}_{\eps}(z)$
and $\phi^{2j}_{\eps}(z)$ only differ by the reflection $r\rightarrow -r$. Hence, it suffices to consider the operators $\phi^{1j}_{\eps}(z)$ henceforth.

In Lemma \ref{lm6}, we will see that $\phi^{1j}_{\eps}(z)$ and $\phi^{2j}_{\eps}(z)$ extend to bounded operators. Hence, we can expect that they converge to the formal limit operators $\phi^{1j}_{0}(z)$ and  $\phi^{2j}_{0}(z)$, respectively, which are defined by the corresponding kernels with $\eps=0$. The next lemma explains the norm bounds in Lemma \ref{lm6}: 

\begin{lemma}\label{lm5}
Let $z>0$. Then the operator $F_z: L^2(\R^2) \rightarrow L^2(\R^2)$ defined by the integral
kernel $K(x,y;x',y')= G_{z}^3\left(x-x',x-y',y-x' \right)$ is bounded with
\begin{align}
\| F_z \| \leq  (2\sqrt{z})^{-1}. \nonumber
\end{align}
\end{lemma}

\begin{proof}
Applying the Schur test yields
\begin{align}
\| F_{z} \| &\leq \sup_{x,y}\left( \int \mspace{-3mu} \d{x'} \, \d{y'} \;G_{z}^3\left(x-x',x-y',y-x' \right) \right)\nonumber \\
&\leq \int \mspace{-3mu} \d{x'} \, \d{y'} \;G_{z}^3\left(0,y',x' \right) = (2\sqrt{z})^{-1}.   \label{SchurTest1}
\end{align}
In the second inequality we first made a substitution and then used that $G_{z}^3(x)$ is decreasing as a function of $|x|$. The integral can be evaluated directly or with the help of \eqref{Greenpartint}.
\end{proof}

The first step in proving that $\phi^{ij}_{0}(z) = \limepsr \phi^{ij}_{\eps}(z)$ for $i \in \{1,2\}$ and $j \in \{3,...,N\}$ is again a reduction to compactly supported potentials. For this purpose, let $\phi^{ij}_{\eps,k}(z) \;(\eps \geq 0)$ be the variant of $\phi^{ij}_{\eps}(z)$ with the potential $V$ replaced by the cutoff potential $V_k$ from \eqref{cutpot}. The next lemma is the analogue of Lemma \ref{lm3}: 

\begin{lemma}\label{lm6}
Let $z>0$, $\eps \geq 0$, $i \in \{1,2\}$ and $j\in \{3,...,N \}$. Then $\phi^{ij}_{\eps}(z)$ extends to a bounded operator from $\widetilde{\HH}$ to $L^2(\R^N)$, which satisfies the norm estimate
\begin{align}
\| \phi^{ij}_{\eps}(z) \| &\leq (2\sqrt{z})^{-1} \| V \|_{L^1}. \label{Abs4} 
\end{align}
Furthermore, for any $k>0$, we have the estimate
\begin{align}
\| \phi^{ij}_{\eps}(z) -\phi^{ij}_{\eps,k}(z) \| &\leq   (\sqrt{2z})^{-1}
\| V \|_{L^1}^{1/2} \| V - V_k \|_{L^1}^{1/2}. \label{Abs5}    
\end{align}
\end{lemma}\medskip

\begin{proof}
Without loss of generality, we may assume that $i=1$. Furthermore, the proofs of \eqref{Abs4} and \eqref{Abs5} are similar.  In both cases we have to estimate the norm of an operator that, for fixed $\underline{P}_j=(P_3,...\widehat{P_j}...,P_N)$, is given by a kernel of the form
$$
      W(r,r') \,G_{z+Q_j}^{3}\left( X_{\eps}^{1j}\right).
$$
Explicitly, we have that $W(r,r')=u(r)v(r')$ in the case of \eqref{Abs4} and $W(r,r')=u(r)v(r')-u_k(r)v_k(r')$ in the case of \eqref{Abs5}. Therefore, we  only demonstrate the desired estimate in the case of \eqref{Abs4}. For fixed $\Psi$, the Cauchy-Schwarz inequality in the $r'$-integration yields 
\begin{align}
&\| \phi^{1j}_{\eps}(z,\underline{P}_j) \Psi \|^2 = \int \mspace{-3mu} \d{r}\,\d{R}\,
\d{x_{j}} \left| \int \mspace{-3mu} \d{r'} \,  \d{R'} \,
 \d{x_{j}'} \; W(r,r')\, G^{3}_{z+Q_j}(X_{\eps}^{1j})\, \Psi(r',R',x_{j}') \right|^2 \nonumber \\
 &\quad\,= \int \mspace{-3mu} \d{r}\,\d{R}\, \d{x_{j}} \left| \int \mspace{-3mu} \d{r'} \,  W(r,r') \int \mspace{-3mu} \d{R'} \,  \d{x_{j}'} \; G^{3}_{z+Q_j}(X_{\eps}^{1j})\, \Psi(r',R',x_{j}') \right|^2 \nonumber \\
 & \quad \leq \int \mspace{-3mu} \d{r}\,\d{R}\, \d{x_{j}}  \left\{  \int \mspace{-3mu} \d{r'} \,  W(r,r') ^2 \right\}  \int \mspace{-3mu} \d{r'} \left| \int \mspace{-3mu} \d{R'} \, \d{x_{j}'} \; G^{3}_{z+Q_j}(X_{\eps}^{1j})\, \Psi(r',R',x_{j}') \right|^2 \nonumber \\
 & \quad \leq  \left\{  \int \mspace{-3mu} \d{r} \, \d{r'} \,  W(r,r')^2  \right\} \cdot \sup_r\left(  \int \mspace{-3mu} \d{r'} \d{R}\, \d{x_{j}} \left| \int \mspace{-3mu} \d{R'} \,
 \d{x_{j}'} \; G^{3}_{z+Q_j}(X_{\eps}^{1j})\, \Psi(r',R',x_{j}') \right|^2 \right)\mspace{-3mu}, \label{phieps1jabs} 
\end{align}
where 
\begin{align}
\int \mspace{-3mu} \d{r} \, \d{r'} \;  W(r,r')^2   = \| V \|_{L^1}^2. \label{IntW^2V1}
\end{align}
In the case of \eqref{Abs5}, the identity \eqref{Videntity} implies that
\begin{align}
\int \mspace{-3mu} \d{r} \, \d{r'} \;  W(r,r')^2   = \| V \|_{L^1}^2  - \| V_k \|_{L^1}^2 \leq 2 \| V \|_{L^1} \| V-V_k \|_{L^1}. \label{IntW^2V2}
\end{align}
The rest of the proof of \eqref{Abs5} is the same as for \eqref{Abs4}. 

We continue estimating the right side of \eqref{phieps1jabs}. For fixed $r,r' \in \R$ and $\underline{P}_j$, the 
sequence of substitutions $R' + \tfrac{\eps (r-r')}{2}\rightarrow R'$, $x_j-\eps r'+ \tfrac{\eps r}{2}\rightarrow x_j$, 
$x_j' - \tfrac{\eps r}{2} \rightarrow x_j'$ leads to
\begin{align}
&\int \mspace{-3mu} \d{R}\,\d{x_{j}} \left| \int \mspace{-3mu} \d{R'} \, \d{x_{j}'} \, G^{3}_{z+Q_j}\left(X^{1j}_{\eps}\right)\,  \Psi(r',R',x_{j}') \right|^2 \nonumber \\
&\quad=\int \mspace{-3mu}\d{R}\,\d{x_{j}} \left| \int \mspace{-3mu} \d{R'} \, \d{x_{j}'} \, G^{3}_{z+Q_j}\left(R-R',R-x_j',x_j-R' \right)\,  \widetilde{\Psi}(r',R',x_j') \right|^2 \nonumber \\
& \quad = \left\|F_{z+Q_j}\widetilde{\Psi}(r',\,\cdot\,) \right\|^2 \label{Trans1}
\end{align}
with the integral operator $F_{z+Q_j}\in \LL(L^2(\R^2))$ from Lemma \ref{lm5} and $\widetilde{\Psi} \in L^2(\R^3)$ defined by
\begin{align}
\widetilde{\Psi}(r',R',x_j'):=\Psi\left(r',R'-  \tfrac{\eps (r-r') }{2}, x_j'+ \tfrac{\eps r}{2} \right). \nonumber
\end{align}

From the estimates \eqref{phieps1jabs}, \eqref{Trans1} and Lemma \ref{lm5} together with the fact that $\|\Psi\|=\|\widetilde{\Psi}\|$, it follows that
\begin{align}
& \| \phi^{1j}_{\eps}(z,\underline{P}_j) \Psi \|^2 \leq (4z)^{-1} \int \mspace{-3mu} \d{r} \, \d{r'} \,  W(r,r')^2\;\|\Psi \|^2, \nonumber
\end{align}
where the right side is independent of $\underline{P}_j$. Hence, $\phi^{1j}_{\eps}(z,\underline{P}_j)$ extends to a bounded operator in $L^2(\R^3)$ and in view of \eqref{IntW^2V1} its norm is bounded by the right side of \eqref{Abs4}. 
This completes the proof of \eqref{Abs4}. The proof of \eqref{Abs5} is similar with the only exception that \eqref{IntW^2V1} has to be replaced by \eqref{IntW^2V2}. 
\end{proof}

After these preparations, we are in the position to prove:

\begin{proposition}\label{lm7}
 Let $z>0$, $i \in \{1,2\}$ and $j\in \{3,...,N \}$. Then $\phi^{ij}_{\eps}(z)$ converges in operator norm to  $\phi_0^{ij}(z)$ as $\eps\to 0$. If $\int \d{r}\, |r|^{2s} \left|V(r)\right| < \infty$ for some $s\in (0,1)$, then $\|\phi^{ij}_{\eps}(z) - \phi^{ij}_0(z)  \| = O(\eps^s)$ as $\eps\to 0$.  
\end{proposition}

\begin{proof} Again, it suffices to consider the case $i=1$. Let us first assume that
\begin{align}
I(V,s):=\int \d{r}\, (1+|r|^{2s}) \left|V(r)\right| < \infty  \label{I(V,s)}
\end{align}
for some $s\in (0,1)$. Observe that $\phi^{1j}_{\eps}(z)- \phi^{1j}_0(z)$ acts pointwise in  $\underline{P}_j$ by the kernel
\begin{align}
u(r) \, \left( G_{z+Q_j}^{3}\left( X_{\eps}\right) - G_{z+Q_j}^{3}\left( X_{0}\right) \right) v(r'),
\end{align}
where $X_{\eps}:=X_{\eps}^{1j}$ for short. Using the Cauchy-Schwarz inequality in the $r'$-integration, we obtain for fixed $\Psi \in L^2(\R^3)$ that
\begin{align}
&\left\| \left(\phi^{1j}_{\eps}(z,\underline{P}_j )-\phi^{1j}_0(z, \underline{P}_j)\right)\Psi\right\|^2 \nonumber \\
 &\;\;\;= \int \mspace{-3mu} \d{r}\, |V(r)|  \int \mspace{-3mu} \d{R}\, \d{x_{j}} \left| \int \mspace{-3mu} \d{r'} \,  v(r') \int \mspace{-3mu} \d{R'} \,  \d{x_{j}'} \; \left(G^{3}_{z+Q_j}(X_{\eps})-G^{3}_{z+Q_j}(X_{0})\right) \Psi(X') \right|^2 \nonumber \\
&\;\;\; \leq  I(V,s) \int \mspace{-3mu} \d{r} \, \d{r'}  \, \frac{|V(r)|}{1+|r'|^{2s}} \int \mspace{-3mu} \d{R}\, \d{x_{j}} \left| \int \mspace{-3mu} \d{R'} \, \d{x_{j}'} \,  \left(G^{3}_{z+Q_j}(X_{\eps})-G^{3}_{z+Q_j}(X_{0})\right) \Psi(X') \right|^2\mspace{-5mu},  \label{phi1jdiffabs}
\end{align}
where $X':=(r',R',x_j')$ for brevity. 

For a further estimate of \eqref{phi1jdiffabs}, we consider for fixed $r,r' \in \R$, $Q_j \geq 0$ and $\eps>0$ the integral operator
$F_{r,r',Q_j,\eps}$ in $L^2\left(\R^2,\d{(R,x_j)}\right)$ that is defined in terms of the kernel
\begin{align}
G^{3}_{z+Q_j}(X_{\eps})-G^{3}_{z+Q_j}(X_{0}) &=
G^{3}_{z+Q_j}\left(R-R'-\tfrac{\eps (r-r')}{2}, R+\tfrac{\eps r}{2}-x_j',x_j-R'-\tfrac{\eps r'}{2} \right) \nonumber\\&\;
\;\;-G^{3}_{z+Q_j}\left( R-R', R-x_{j}', x_{j}-R'\right).\nonumber
\end{align}
We are going to estimate $\|F_{r,r',Q_j,\eps} \|$ with the help of a Schur test. To this end, we introduce for $\eps\geq 0$ the intermediate point
$$
       X_{\eps,0}:=\left(R-R'-\tfrac{\eps (r-r')}{2}, R-x_{j}', x_{j}-R'\right).
$$
Using the properties of the Green's function (see Lemma \ref{lmA2}), 
\begin{align}
&\sup_{R,x_j}\left( \int \mspace{-3mu} \d{R'}\,\d{x_{j}'} \left| G^{3}_{z+Q_j}(X_{\eps})-G^{3}_{z+Q_j}(X_{0}) \right|  \right) \nonumber \\ 
&  \leq  \sup_{R,x_j}\left( \int \mspace{-3mu} \d{R'}\,\d{x_{j}'} \left| G^{3}_{z}(X_{\eps})-G^{3}_{z}(X_{\eps,0}) \right|  \right) 
+ \sup_{R,x_j}\left( \int \mspace{-3mu} \d{R'}\,\d{x_{j}'} \left| G^{3}_{z}(X_{\eps,0})-G^{3}_{z}(X_{0}) \right|  \right) \nonumber \\
& \leq \sup_{R,x_j}\left( \int \mspace{-3mu} \d{R'}\,\d{x_{j}'} \left| G^{3}_{z}\left( 0, R-x_{j}' +\tfrac{\eps r}{2}, x_{j}-R' -\tfrac{\eps r'}{2} \right) - G^{3}_{z}\left( 0, R-x_{j}', x_{j}-R'\right) \right|  \right) \nonumber  \\
& \quad+ \sup_{R,x_j}\mspace{-2mu}\left( \int \mspace{-4mu} \d{R'}\,\d{x_{j}'} \left| G^{3}_{z}\mspace{-2mu} \left(\mspace{-3mu} R\mspace{-1mu}-\mspace{-1mu}R' - \tfrac{\eps(r-r') }{2}, R-x_{j}',\mspace{-1mu} 0\mspace{-2mu}\right)\mspace{-3mu}-\mspace{-1mu} G^{3}_{z}\mspace{-2mu} \left(R\mspace{-1mu}-\mspace{-1mu}R', R-x_{j}', \mspace{-1mu}0\mspace{-2mu}\right) \right|  \right) \nonumber \\
& \leq  \int \mspace{-3mu} \d{R'}\,\d{x_{j}'} \left| G^{3}_{z}\left( 0, x_{j}' + \tfrac{\eps r}{2}, R' -\tfrac{\eps r'}{2}\right) - G^{3}_{z}\left( 0, x_{j}', R'\right) \right|   \nonumber  \\
& \quad+\;  \int \mspace{-3mu} \d{R'}\,\d{x_{j}'} \left|      \,G^{3}_{z}\mspace{-2mu} \left(R' - \tfrac{\eps (r-r') }{2}, x_{j}', 0\right)\mspace{-3mu}-\mspace{-1mu} G^{3}_{z}\mspace{-2mu} \left(R',x_{j}', 0\right) \right|, \label{SchurTest2}
\end{align} 
where we substituted $(R-x_j') \rightarrow x_j', (x_j-R') \rightarrow R'$ in the first and $R-R' \rightarrow R', (R-x_j') \rightarrow x_j'$ in the second integral. From \eqref{SchurTest2} and a similar estimate with the roles of $(R,x_j)$ and $(R',x'_j)$ interchanged, we conclude, using the Schur test, that
\begin{align}
\| F_{r,r',Q_j,\eps} \| \leq 2 \sup_{\left|y\right| \leq |r| + |r'|}\left(\; \int\limits_{\R^{2}} \mspace{-3mu} \d{x}\, \left|G_z^3(x+\eps y/2,0) - G_z^3(x,0)  \right| \right). \label{Fnormabs}
\end{align}
Hence, Lemma \ref{lmA3} implies that $\| F_{r,r',Q_j,\eps} \|^2 \leq C \eps^{2s}(|r|^{2s} + |r'|^{2s})$ for some constant $C=C(s,z)>0$, which does not depend on $r,r',Q_j$ and $\eps$. 
Using this in \eqref{phi1jdiffabs} results in
 \begin{align}
\left\| \left(\phi^{1j}_{\eps}(z,\underline{P}_j )-\phi^{1j}_0(z, \underline{P}_j)\right)\Psi\right\|^2 &\leq  I(V,s) \int \mspace{-3mu} \d{r} \, \d{r'}  \, \frac{|V(r)|}{1+|r'|^{2s}} \left\| F_{r,r',Q_j,\eps}\Psi(r',\,\cdot\,)\right\|^2 \nonumber \\
& \leq C \eps^{2s} I(V,s)^2 \, \|\Psi \|^2, \nonumber
\end{align}
where $I(V,s)$ is defined by \eqref{I(V,s)}. 
As the right side is independent of $\underline{P}_j$, this proves the second part of the lemma. 

If the additional assumption \eqref{I(V,s)} is not satisfied for some $s>0$, then the proposition follows, by an approximation argument, from Lemma \ref{lm6}  and from the fact that $I(V_k,s)<\infty$ for finite $k$.
\end{proof}


\subsection{The limit of $\phi_{\eps}^{ij}(z) \; (3 \leq i<j\leq N)$} 
\label{sec43}

So far, we have discussed all the operators $\phi_{\eps}^{ij}(z)$ that occur in the case $N\leq 3$.
If $N> 3$, there are, in addition, the contributions from $\phi^{ij}_{\eps}(z)$ with $\{i,j\} \subset \{3,\ldots\}$. Given such a pair $(i,j)$, we determine the kernel of $\phi^{ij}_{\eps}(z)$ from \eqref{Defphiepsij}. 
After evaluating the $\delta$-distributions in $x_i'$ and $x_j'$ followed by the substitutions  $x_1' \rightarrow x_i'$ and $x_2' \rightarrow x_j'$, we see that $\phi^{ij}_{\eps}(z)$ has the kernel
\begin{align}
 u(r)\,  G_{z}^{N} \mspace{-2mu}\left(X^{ij}_{\eps},x_3\mspace{-2mu}-\mspace{-2mu}x_3',...\widehat{x_{i}\mspace{-2mu}-\mspace{-2mu}x_{i}'}...\widehat{x_{j}\mspace{-2mu}-\mspace{-2mu}x_{j}'}...,x_{N}\mspace{-2mu}-\mspace{-2mu}x_N'\right)\mspace{-2mu}v\mspace{-2mu}\left(r'\right), \label{phiepsijkern}
\end{align}
where we made use of the shorthand notation
\begin{align}
X^{ij}_{\eps}:=\left(R - \tfrac{\eps r}{2}-x_i', R + \tfrac{\eps r}{2} -x_j', x_{i}-R'+ \tfrac{\eps r'}{2}, x_{j}-R'- \tfrac{\eps r'}{2} \right).\label{DefX^ij}
\end{align}
After a Fourier transform in $(x_1,...\widehat{x_{i}}...\widehat{x_{j}}...,x_N)$, property (vi) of Lemma \ref{lmA1} implies that $\phi^{ij}_{\eps}(z)$ acts pointwise in the associated momentum variable $\underline{P}_{ij}=(P_3,...\widehat{P_{i}}...\widehat{P_{j}}...,P_{N})$ by an operator  $\phi^{ij}_{\eps}(z,\underline{P}_{ij})$ with kernel
\begin{align}
 u(r)\,  G_{z+Q_{ij}}^{4} \mspace{-2mu}\left(X^{ij}_{\eps}\right)\mspace{-2mu}v\mspace{-2mu}\left(r'\right), \qquad \quad Q_{ij}:=\sum_{\substack{l= 3 \\ l \notin \{i,j\}}}^N \mspace{-5mu} P_l^2. \label{phiepsijkernP}
\end{align}
In Lemma \ref{lm9}, we will see that $\phi^{ij}_{\eps}(z)$ extends to a bounded operator  from $\widetilde{\HH}$
to $L^2(\R^N)$. Moreover, this is still true for the (formal) limit operator $\phi^{ij}_{0}(z)$ which, for fixed $\underline{P}_{ij} \in \R^{N-4}$, is defined by the kernel \eqref{phiepsijkernP} with $\eps=0$. The norm bound in Lemma \ref{lm9} will be a consequence of the following lemma:

\begin{lemma}\label{lm8}
Let $z>0$. Then the operator $B_z: L^2(\R^3) \rightarrow L^2(\R^3)$ defined by the integral
kernel $B(w,x,y;w',x',y')= G_{z}^4\left(w-x',w-y',x-w',y-w'\right)$ is bounded with
\begin{align}
\| B_z \| \leq  (2\sqrt{z})^{-1}. \nonumber
\end{align}
\end{lemma}

\begin{proof}
The Schur test and the substitutions $w-x' \rightarrow x'$, $w-y' \rightarrow y'$ , $x-w'\rightarrow w'$ yield
\begin{align}
\| B_z \| &\leq \sup_{w,x,y}\left( \int \mspace{-3mu}  \d{w'} \, \textup{d}x' \, \d{y'} \; G^{4}_{z}\mspace{-1mu}\left(w-x',w-y',x-w',y-w'\right)  \right)\nonumber \\
&\leq\sup_{x,y}\left( \int \mspace{-3mu}  \d{w'} \, \textup{d}x' \, \d{y'}   \; G^{4}_{z}\mspace{-1mu}\left(x',y',w',w'-x+y  \right)\right). 
\end{align}
As $G_{z}^4(x)$ is decreasing as a function of $|x|$ (see Lemma \ref{lmA1} (v)), we can continue estimating
\begin{align}
\| B_z \| \leq \int \mspace{-3mu}  \d{w'} \, \textup{d}x' \, \d{y'}   \; G^{4}_{z}\mspace{-1mu}\left(x',y',w',0  \right)=G^1_{z}(0)=(2\sqrt{z})^{-1},   \label{B_znormabs}
\end{align}
where we used \eqref{Greenpartint} to evaluate the integral.
\end{proof}

As in the previous sections, the first step in proving convergence is a reduction to potentials with compact support. For this reason, we introduce for $k>0$ and $\eps \geq 0$ the operator $\phi^{ij}_{\eps,k}(z)$ which is nothing but the operator $\phi^{ij}_{\eps}(z)$ with the potential $V$ replaced by the cutoff potential $V_k$ from \eqref{cutpot}. Hence, for fixed $\underline{P}_{ij}$, the kernel of $\phi^{ij}_{\eps,k}(z,\underline{P}_{ij})$ is given by \eqref{phiepsijkernP} with $u_k$ instead of $u$ and $v_k$ instead of $v$. The next lemma is the analogue of Lemma \ref{lm6}:

\begin{lemma}\label{lm9}
Let $z>0$, $\eps\geq 0$ and $3 \leq i < j \leq N$. Then $\phi^{ij}_{\eps}(z)$ extends to a bounded operator from $\widetilde{\HH}$ to $L^2(\R^{N})$, which satisfies the norm estimate
\begin{align}
\| \phi^{ij}_{\eps}(z) \| &\leq (2\sqrt{z})^{-1} \, \| V \|_{L^1}. \label{Abs6}
\end{align}
Furthermore, for any $k>0$, we have the estimate
\begin{align}
\|  \phi^{ij}_{\eps}(z) - \phi^{ij}_{\eps,k}(z) \| &\leq (\sqrt{2z})^{-1}
\| V \|_{L^1}^{1/2} \| V - V_k \|_{L^1}^{1/2}. \label{Abs7}    
\end{align}
\end{lemma}

\begin{proof} 
The proof follows the line of arguments in the proof of Lemma~\ref{lm6} with very few adjustments. The role of $F_z$ in the proof of Lemma \ref{lm6} is now played by $B_z \in \LL(L^2(\R^3))$, for which we use the estimate from Lemma \ref{lm8}.
\end{proof}

As the last step before the proof of Theorem \ref{theo1}, we show convergence of $\phi^{ij}_{\eps}(z)$:

\begin{proposition}\label{lm10}
Let $z>0$ and $3 \leq i < j \leq N$. Then $\phi^{ij}_{\eps}(z)$ converges in operator norm to  $\phi_0^{ij}(z)$ as $\eps\to 0$.  If $\int \d{r}\, |r|^{2s} \left|V(r)\right| < \infty$ for some $s\in(0,1)$, then $\|\phi^{ij}_{\eps}(z) - \phi^{ij}_0(z)  \| = O(\eps^s)$ as $\eps\to 0$.
\end{proposition}

\begin{proof} 
As in the proof of Proposition~\ref{lm7}, we first assume that the additional condition \eqref{I(V,s)} is satisfied for some $s\in (0,1)$. In the following, we will make use of the shorthand notations  $X_{\eps}:=X_{\eps}^{ij}$ for $\eps \geq 0$ and $q:=Q_{ij}$.
Note that \eqref{phiepsijkernP} implies that $\phi^{ij}_{\eps}(z)-\phi_0^{ij}(z)$ acts pointwise in $\underline{P}_{ij}$ by an operator with kernel
\begin{align}
u(r) \left(G_{z+q}^{4}\left(X_{\eps}\right)-G_{z+q}^{4}\left(X_{0}\right) \right) v(r'). \label{phiijdiffkern}
\end{align}
Hence, similarly to \eqref{phi1jdiffabs}, Cauchy-Schwarz inequality yields that

\begin{align}
&\left\| \left(\phi^{ij}_{\eps}(z,\underline{P}_{ij} )-\phi^{ij}_0(z, \underline{P}_{ij})\right)\Psi\right\|^2 \nonumber \\
&\;\leq \int \mspace{-3mu} \d{r} \, \d{r'}  \; \frac{ I(V,s) \, |V(r)|}{1+|r'|^{2s}} \int \mspace{-3mu} \d{R}\, \d{x_i} \, \d{x_{j}} \left| \int \mspace{-3mu} \d{R'} \, \d{x_{i}'} \, \d{x_{j}'} \left(G^{4}_{z+q}(X_{\eps})-G^{4}_{z+q}(X_{0})\right) \Psi(X') \right|^2\mspace{-5mu},  \label{phiijdiffabs}
\end{align}
where $X':=(r',R',x_i',x_j')$ for short. 

For a further estimate of \eqref{phiijdiffabs}, we analyse for fixed $r,r' \in \mathbb{R}$, $q \geq 0$ and $\eps>0$ the integral operator $B_{\eps,r,r'\mspace{-3mu},\mspace{1mu} q}$ in $L^2(\R^3,\d{(R,x_i,x_j)})$ with kernel
\begin{align}
G^{4}_{z+q}(X_{\eps})-G^{4}_{z+q}(X_{0})&=G^{4}_{z+q}\mspace{-3mu}\left(R - \tfrac{\eps r}{2}-x_i', R + \tfrac{\eps r}{2} -x_j', x_{i}-R'+ \tfrac{\eps r'}{2}, x_{j}-R'- \tfrac{\eps r'}{2}\right) \nonumber \\
&\;\;\;-G^{4}_{z+q}\left(R -x_i', R  -x_j', x_{i}-R', x_{j}-R'\right). \nonumber
\end{align}
In order to estimate the operator norm of $B_{\eps,r,r'\mspace{-2mu},q}$,
we introduce for $\eps \geq 0$ the intermediate point
\begin{align}
X_{\eps,0}:=\left(R -x_i', R -x_j', x_{i}-R', x_{j}-R'- \tfrac{\eps r'}{2} \right).\nonumber
\end{align}
Properties of the Green's function (see Lemma \ref{lmA2}) in combination with substitutions yield the estimate
\begin{align}
& \sup_{R,x_i,x_j}\mspace{-3mu}\left( \int \mspace{-3mu} \d{R'}\,\d{x_{i}'}\,\textup{d}x_{j}' \left| G^{4}_{z+q}(X_{\eps})-G^{4}_{z+q}(X_{0}) \right|  \right) \nonumber \\ 
&\leq \sup_{R,x_i,x_j}\mspace{-3mu}\left( \int \mspace{-3mu} \d{R'}\,\d{x_{i}'}\,\textup{d}x_{j}' \left|  G^{4}_{z}(X_{\eps})-G^{4}_{z}(X_{\eps,0})  \right| +  \int \mspace{-3mu} \d{R'}\,\d{x_{i}'}\,\textup{d}x_{j}' \left| G^{4}_{z}(X_{\eps,0})-G^{4}_{z}(X_{0}) \right|  \right) \nonumber \\
&\leq \sup_{x_i}\mspace{-3mu}\left( \int \mspace{-3mu} \d{R'}\,\d{x_{i}'}\,\textup{d}x_{j}' \left| G^{4}_{z}(x_i'-\tfrac{\eps r}{2}, x_j'+\tfrac{\eps r}{2},x_i-R'+\tfrac{\eps r'}{2}, 0)-G^{4}_{z}(x_i',x_j',x_i-R', 0) \right| \right) \nonumber\\
&\;\;\;+\mspace{-3mu}\sup_{x_j}\mspace{-3mu}\left( \int \mspace{-3mu} \d{R'}\,\d{x_{i}'}\,\textup{d}x_{j}' \left| G^{4}_{z}(x_i',x_j',0,x_j-R'-\tfrac{\eps r'}{2})\mspace{-1mu}-\mspace{-1mu}G^{4}_{z}(x_i',x_j',0, x_j-R') \right| \right) \nonumber \\
&\leq \int \mspace{-3mu} \d{R'}\,\d{x_{i}'}\,\textup{d}x_{j}' \left| G^{4}_{z}(x_i'-\tfrac{\eps r}{2}, x_j'+\tfrac{\eps r}{2},R'+\tfrac{\eps r'}{2}, 0)-G^{4}_{z}(x_i',x_j',R', 0) \right| \nonumber\\
&\;\;\;+\int \mspace{-3mu} \d{R'}\,\d{x_{i}'}\,\textup{d}x_{j}' \left| G^{4}_{z}(x_i',x_j',R'-\tfrac{\eps r'}{2},0)\mspace{-1mu}-\mspace{-1mu}G^{4}_{z}(x_i',x_j',R',0) \right| \nonumber \\
& \leq 2 \sup_{\left|y\right| \leq |r| + |r'|}\left(\; \int\limits_{\R^{3}} \mspace{-3mu} \d{x}\, \left|G_z^4(x+\eps y,0) - G_z^4(x,0)  \right| \right).
\label{SchurTest3}
\end{align} 
By the Schur test and by  Lemma \ref{lmA3}, estimate \eqref{SchurTest3} and a similar estimate with the roles of $(R,x_i,x_j)$ and $(R',x'_i,x'_j)$ interchanged, imply that $\| B_{\eps,r,r',q} \|^2 \leq C \eps^{2s}(|r|^{2s} + |r'|^{2s})$ for some constant $C=C(s,z)>0$, which is independent of $\eps, r,r'$ and $q$. Hence, it follows from 
\eqref{phiijdiffabs} that
 \begin{align}
\left\| \left(\phi^{ij}_{\eps}(z,\underline{P}_{ij} )-\phi^{ij}_0(z, \underline{P}_{ij})\right)\Psi\right\|^2 &\leq  I(V,s) \int \mspace{-3mu} \d{r} \, \d{r'}  \, \frac{|V(r)|}{1+|r'|^{2s}} \left\| B_{\eps,r,r',q}\Psi(r',\,\cdot\,)\right\|^2 \nonumber \\
& \leq C \eps^{2s} I(V,s)^2 \, \|\Psi \|^2, \nonumber
\end{align}
where $I(V,s)$ is defined by \eqref{I(V,s)}. Due to the fact that this bound is uniform in $\underline{P}_{ij}$,  the second part of the lemma follows. 

In the general case, where the additional assumption \eqref{I(V,s)} is not satisfied for some $s>0$, the usual approximation argument, now based on Lemma \ref{lm9}, proves the convergence of $\phi^{ij}_{\eps}(z)$ to $\phi^{ij}_{0}(z)$.
\end{proof}

\section{Proof of the main results} \label{sec5}

\begin{proof}[Proof of Theorem~\ref{theo1}]
Recall from \eqref{phiepsdecomp} that $\phi_{\eps}(z) \in \LL(\widetilde{\HH})$ can be decomposed as
\begin{align}\label{phiepsdecomp2}
    \phi_{\eps}(z)=\sum_{1 \leq i < j \leq N} \phi^{ij}_{\eps}(z). 
\end{align}
According to \eqref{Abs2}, \eqref{Abs4} and \eqref{Abs6}, all operators $\phi^{ij}_{\eps}(z): \widetilde{\HH} \rightarrow L^2(\R^N)$ are bounded with $\|\phi^{ij}_{\eps}(z)\| \rightarrow 0$ uniformly in $\eps \geq 0$ as $z \rightarrow \infty$. As $g=\limepsr g_{\eps}$ exists, we see that there exists $C>0$, independent of $\eps \in (0,\eps_0)$, such that $1-g_{\eps}\phi_{\eps}(z)$ is invertible for all $z>C$. Now, Proposition~\ref{lm4}, Proposition~\ref{lm7}, and Proposition~\ref{lm10} imply that
\begin{align} \label{phiinvconv2}
     \phi_{\eps}(z) \rightarrow \phi(z):= \sum_{1 \leq i < j \leq N} \phi_{0}^{ij}(z)
\end{align} 
as $\eps\to 0$, and by the argument above, $1-g\phi(z)$ is invertible for $z>C$. 
We conclude that $\limepsr g_{\eps}\left(1 - g_{\eps}\phi_{\eps}(z)\right)^{-1} =g\left(1 - g\phi(z)\right)^{-1}$ for all $z>C$.

Given the invertibility of $1-g_{\eps}\phi_{\eps}(z)$ for $z>C$ and $\eps\in (0,\eps_0)$, we know from 
Theorem~\ref{thm:Krein} that $z\in\rho(H_\eps)$ and
\begin{align}\label{Krein2}
(H_{\eps}+z)^{-1} = (H_0 + z)^{-1} + g_{\eps} \left(A_{\eps}(H_0 + \overline{z})^{-1}\right)^{*} \left(1-g_{\eps} \phi_{\eps}(z)\right)^{-1} JA_{\eps} (H_0 + z)^{-1},
\end{align}
where, by \eqref{LimitS(z)}, $A_{\eps}(H_0+z)^{-1} \rightarrow S(z)$. Taking the limit $\eps\to 0$ on the right hand side of \eqref{Krein2}, we conclude that $(H_{\eps}+z)^{-1} \to R(z)$, where
\begin{equation} \label{DefR(z)}
          R(z):= (H_0 + z)^{-1} + g \, S(\overline{z})^{*} \,\left(1-g\phi(z)\right)^{-1}JS(z).
\end{equation}

To see that $R(z)=(H+z)^{-1}$ we use Corollary~\ref{cor:VfreeR}, below. By this corollary, the operator $R(z)$ depends on $\alpha=g\int V(r) \,\d{r}$ only, but not on the particular choice of $V$. This means that $R(z)=\lim_{\eps\to 0}(H_\eps'+z)^{-1}$, where $H_\eps'$ is defined in terms of a compactly supported potential $V'$ with same integral as $V$. By Corollary~\ref{strongRlim}, $H_\eps'\to H$ in the strong resolvent sense and hence $R(z)=(H+z)^{-1}$. This proves Equation~\eqref{Hres} for $z>C$. By a general result on equations of this form (see Theorem 2.19 in \cite{CFP2018}), it follows that $1-g\phi(z)$ is invertible and that Equation~\eqref{Hres} holds for all $z\in \rho(H_0)\cap\rho(H)$.

Finally, let us assume that $\int \d{r}\, |r|^{2s} \left|V(r)\right| < \infty$ and $|g_{\eps}-g|=O(\eps^s)$ for some $s\in(0,1)$. Then, by Propositions~\ref{lm2}, \ref{lm4}, \ref{lm7}, and \ref{lm10}, 
\begin{align}
   \left\|A_{\eps}(H_0+z)^{-1}-S(z)\right\|&=O(\eps^s)\\
   \left\|g_{\eps} \left(1 - g_{\eps}\phi_{\eps}(z)\right)^{-1}  - g\left(1-g\phi(z)\right)^{-1} \right\| &=O(\eps^s)
\end{align}
as $\eps\to 0$. Using this in \eqref{Krein2} shows that $\|(H_{\eps}+z)^{-1} - (H+z)^{-1} \| = O(\eps^s)$ for $z>C$, and hence by Lemma~\ref{theo3}, $\|(H_{\eps}+z)^{-1} - (H+z)^{-1}  \| = O(\eps^s)$ for all $z \in \rho(H)$. This concludes the proof of Theorem \ref{theo1}.
\end{proof}

\bigskip
Finally, we show that the potential $V$ may be integrated out in the expression \eqref{DefR(z)} for the resolvent of $H$. To this end, we introduce an auxiliary Hilbert space $\widetilde{\HH}_{\textup{red}}$ by \eqref{HTilde} and
$$
    \widetilde{\HH} =L^2_{\text{ev}}(\R,\d{r}) \otimes \widetilde{\HH}_{\textup{red}}.       
$$
By inspection of the defining equations \eqref{Tepskern}, \eqref{phi12kern}, \eqref{phieps1jkernP}, and \eqref{phiepsijkernP}, the operators $S(z) \in \mathscr{L}(\HH, \widetilde{\HH})$ and $\phi(z) \in \mathscr{L}(\widetilde{\HH})$ factor as follows:
\begin{align}
    S(z)\psi&= v \otimes (\widetilde{S}(z)\psi), \label{DefS(z)Tilde2} \\
    \phi(z)&= \Ket{u}\Bra{v} \otimes \widetilde{\phi}(z), \label{DefphiTilde2}
\end{align}
where $\widetilde{S}(z) \in \mathscr{L}(\HH,\widetilde{\HH}_{\textup{red}})$ and $\widetilde{\phi}(z) \in \mathscr{L}(\widetilde{\HH}_{\textup{red}})$ are bounded operators, which do not depend on $V$. Explicitly, it follows from \eqref{LimitS(z)} that 
\begin{align}
          \widetilde{S}(z)=\sqrt{\frac{N(N-1)}{2}}\;\widetilde{T_0}(z)\KK, \label{S(z)Tilde}
\end{align}
where the action of $\widetilde{T_0}(z)$ is pointwise in $\underline{P}=(P,P_3,...,P_N)$ and given by 
\begin{align}
     \widetilde{T_0}(z, \underline{P})= \Scalefrac{1}{2} \Bra{G_{\frac{1}{2}(z+Q)}}: L^2(\R) \rightarrow \C. 
\end{align}
In view of \eqref{DefphiTilde2}, we have that $\|\phi(z)\| =   \left\|\widetilde{\phi}(z) \right\| \|V\|_{L^1}$, so Lemma \ref{lm3}, Lemma \ref{lm6} and Lemma \ref{lm9} show that
\begin{align}
\left\|\widetilde{\phi}(z) \right\| \leq \frac{K}{\sqrt{z}} 
\end{align}
for some constant $K>0$. Therefore, $1 - \alpha \widetilde{\phi}(z)$ is invertible for large enough $z>0$ and with the help of \eqref{DefphiTilde2} and $\alpha = g\int V(r) \,\d{r}=g\Braket{v|u}$ it is straightforward to verify that 
\begin{align}
g \left(1-g\phi(z)\right)^{-1} = g  + g^2 \Ket{u} \Bra{v} \otimes \widetilde{\phi}(z)\left(1- \alpha \, \widetilde{\phi}(z) \right)^{-1}. \label{phiTildeinv}
\end{align}
Inserting this in \eqref{DefR(z)} gives:

\begin{kor}\label{cor:VfreeR}
With the above notations, for all $z\in \rho(H)\cap\rho(H_0)$,
\begin{equation}\label{R(z)new}
   (H+z)^{-1} = (H_0 + z)^{-1} + \alpha \, \widetilde{S}(\overline{z})^{*} \left(1  - \alpha  \widetilde{\phi}(z) \right)^{-1} \widetilde{S}(z).
\end{equation}
\end{kor}

The following lemma, a variant of Lemma 2.6.1 in \cite{Davies1995}, is used in the proof of Theorem~\ref{theo1}. 
\begin{lemma}\label{theo3}
Let $R_\eps(z)=(H_\eps+z)^{-1}$ and $R(z)=(H+z)^{-1}$. Then for any $z,z_0\in \rho(H)$ there exists a constant $C_z\in \R$ such that for $\eps>0$ small enough
$$
  \|R_\eps(z) - R(z)\| \leq (1 + |z-z_0|C_z)^2 \|R_\eps(z_0) - R(z_0)\|.
$$
\end{lemma}

\begin{proof}
The norm resolvent convergence $H_\eps\to H$ implies that for given $z,z_0\in \rho(H)$ there exists $\eps_0>0$ such that for $\eps<\eps_0$ we have $z,z_0\in \rho(H_\eps)$ and 
$C_z:= \sup_{\eps<\eps_0}\|R_\eps(z)\|<\infty$.
Let $\eps,\delta \in (0,\eps_0)$. By the first resolvent identity
\begin{equation}\label{first-res}
    R_\eps(z) = R_\eps(z_0)F_\eps(z) = F_\eps(z) R_\eps(z_0),
\end{equation}
where $F_\eps(z)=1+(z_0-z)R_\eps(z)$ has the norm bound $\|F_\eps(z)\| \leq 1+ |z_0-z|C_z$.
The second resolvent identity $R_\eps(z) -R_\delta(z) = R_\eps(z)(H_\delta-H_\eps)R_\delta(z)$ and \eqref{first-res} imply that
$$
   R_\eps(z) -R_\delta(z) = F_\eps(z)\big(R_\eps(z_0)-R_\delta(z_0)\big)F_\delta(z).   
$$
After taking norms of both sides and the limit $\delta\to 0$, the desired estimate follows.
\end{proof}

\appendix

\setcounter{theorem}{0}
\setcounter{equation}{0}

\section{Properties of the Green's function}  
\label{AppendixA}

This appendix collects facts and estimates on the Green's function of $-\Delta + z: H^2(\R^d) \rightarrow L^2(\R^d)$.

For $d \in \N$ and $z \in \C$ with $\Rea(z)>0$, let the function $G^d_{z}: \R^d \rightarrow \C$ be defined by
\begin{align} \label{DefGreen}
     G^d_{z}(x) := \int\limits_{0}^{\infty} \textup{dt}\; (4\pi t)^{-d/2}  \exp\left( -\Scalefrac{x^2}{4t} - zt\right).
\end{align}
Note that $ G^d_{z}$ has a singularity at $x=0$ unless $d=1$. The proof of the following lemma is left as an exercise to the reader.

\begin{lemma}\label{lmA1}
Let $d \in \N$ and $z \in \C$ with $\Rea(z)>0$. Then $G^d_{z}$ defined by \eqref{DefGreen} has the following properties:
\begin{enumerate}[label=(\roman*)]
\item $G^d_{z} \in L^1(\R^d)$ and $\|G^d_{z} \|_{L^1}\leq \Rea(z)^{-1}$ 
\item The Fourier transform of $G^d_{z}$ is given by $\widehat{G^d_{z}}(p)=(2\pi)^{-d/2}(p^2+z)^{-1}$
\item $G^d_{z}$ is the Green's function of $-\Delta + z$, i.e. $(-\Delta + z)^{-1}f = G^d_{z} \ast f$ holds for all $f \in L^2(\R^d)$
\item  $G^d_{z} \in L^2(\R^d)$ if and only if $d \in \{1,2,3\}$ 
\item $G^d_{z}$ is spherically symmetric, i.e. $G^d_{z}$ only depends on $\left|x\right|$, and  for $z \in (0,\infty)$ it is positive and 
strictly monotonically decreasing both in $\left|x\right|$ and $z$
\item  Let $d_1,d_2\in \N$ with $d_1+d_2=d$ and let $x=(x_1,x_2) \in \R^{d_1} \times \R^{d_2}$. If $x_1 \neq 0$ or $d_1=1$, then $G^d_z(x_1,\cdot) \in L^1(\R^{d_2})$ and the Fourier transform is
\begin{align}
\widehat{G^d_{z}}(x_1,p_2)=(2\pi)^{-d_2/2} \; G^{d_1}_{z+p_2^2}(x_1). \label{GreenpartP}
\end{align}
In particular, we have that
\begin{align}
 \int\limits_{\R^{d_2}}^{} \textup{d}x_2\; G^d_{z}(x_1,x_2) = G_{z}^{d_1}(x_1).  \label{Greenpartint}
\end{align}
\end{enumerate}
\end{lemma} 


The following lemma is one of our main tools for estimating differences of integral
operators that depend on $G^{d}_z$:

\begin{lemma}\label{lmA2}
Let $d \in \N$ and $z \in \C$ with $\Rea(z)>0$. Then, for all $x, \tilde{x} \in \R^d$ and $Q \geq 0$, it holds that \vspace*{-3mm}
\begin{align}
\left| G^d_{z+Q}(x) - G^d_{z+Q}(\tilde{x}) \right| &\leq \left| G^d_{\Rea(z)}(x) - G^d_{\Rea(z)}(\tilde{x}) \right|.  \label{Greenabs1} 
\end{align}
Similarly, if $d=d_1+d_2$ for some $d_1,d_2 \in \N$, then for all $ x_1, y_1 \in \R^{d_1}$ and all $x_2 \in \R^{d_2}$,
\begin{align}
\left| G^d_{z}(x_1,x_2) - G^d_{z}(y_1,x_2) \right| &\leq \left| G^d_{\Rea(z)}(x_1,0) - G^d_{\Rea(z)}(y_1,0) \right|. \label{Greenabs2}
\end{align}
\end{lemma} 

\begin{proof}
To prove \eqref{Greenabs2} we may assume, without loss of generality, that $\left|x_1\right| \leq \left|y_1\right|$. Then,
\begin{align}
\left| G^d_{z}(x_1,x_2) - G^d_{z}(y_1,x_2) \right| &= \left|\int\limits_{0}^{\infty} \textup{dt}\; 
(4\pi t)^{-d/2} \left( \exp\left( -\Scalefrac{x_1^2}{4t} \right) - \exp\left( -\Scalefrac{y_1^2}{4t} \right)  \mspace{-3mu} \right) \exp\left(  -\Scalefrac{x_2^2}{4t} -zt\right) \right| \nonumber \\
&\leq  \int\limits_{0}^{\infty} \textup{dt}\; 
(4\pi t)^{-d/2} \left( \exp\left( -\Scalefrac{x_1^2}{4t} \right) - \exp\left( -\Scalefrac{y_1^2}{4t} \right) \mspace{-3mu} \right)\left| \exp\left(  -\Scalefrac{x_2^2}{4t} -zt\right) \right| \nonumber \\
&\leq \int\limits_{0}^{\infty} \textup{dt}\; 
(4\pi t)^{-d/2} \left( \exp\left( -\Scalefrac{x_1^2}{4t} \right) - \exp\left( -\Scalefrac{y_1^2}{4t} \right) \mspace{-3mu} \right) \exp\left(  \mspace{-2mu} -\Rea(z)t  \right) \nonumber \\
&= \left| G^d_{\Rea(z)}(x_1,0) - G^d_{\Rea(z)}(y_1,0) \right|. \nonumber
\end{align}
The proof of \eqref{Greenabs1} is very similar.
\end{proof}

\begin{lemma} \label{lmA3}
Let $d \geq 2$, $s \in (0,1)$ and $z \in \C$ with $\Rea(z)>0$. Then there exists a constant $C=C(s,z)>0$ such that 
\begin{equation}\label{G0L1}
 \forall y \in \R^{d-1}: \qquad      \int\limits_{\R^{d-1}} \mspace{-10mu} \d{x}\, \left|G_z^d(x+y,0) - G_z^d(x,0)  \right|\leq C |y|^s. 
\end{equation}
\end{lemma}

\begin{proof} Since $G_z^d(\cdot,0) \in L^1(\R^{d-1})$ by Lemma \ref{lmA1}, (vi), the left side of \eqref{G0L1} is bounded uniformly in $y \in \R^{d-1}$. So it remains to prove \eqref{G0L1} for $|y| \leq 1$, and to this end it suffices to show that there exists a constant $C=C(z)>0$ such that
\begin{align}\label{G0L1alt}
\forall y \in \R^{d-1}: \qquad    \int\limits_{\R^{d-1}} \mspace{-10mu} \d{x}\, \left|G_z^d(x+y,0) - G_z^d(x,0)  \right|\leq C (1+ \left|\ln(|y|)\right|) \,|y|. 
\end{align}
Using the integral representation \eqref{DefGreen} for $G_z^d$ and making the substitution $x/\sqrt{4t}\to x$, we find
\begin{align*}
&\int\limits_{\R^{d-1}} \mspace{-10mu} \d{x}\, \left|G_z^d(x+y,0) - G_z^d(x,0)\right|
\leq \int\limits_{0}^{\infty}\d{t}\, \frac{e^{-ut}}{2 \pi^{d/2} t^{1/2}} 
\int\limits_{\R^{d-1}} \mspace{-10mu} \d{x} \left|\exp\left( -\left(x+\Scalefrac{y}{\sqrt{4t}} \right)^2\right) - \exp\left(-x^2\right) \right|,
\end{align*}
where $u=\Rea(z)$. By applications of triangle inequality and the fundamental theorem of calculus, respectively,
$$
   \int\limits_{\R^{d-1}}\mspace{-10mu}\d{x}\left|\exp\left(-\left(x+\Scalefrac{y}{\sqrt{4t}}\right)^2\right) - \exp\left(-x^2\right)\right|   \leq C \min\left(1, \frac{|y|}{\sqrt{t}}\right).
$$
Since 
\begin{eqnarray*}
   \int\limits_{0}^{\infty}\d{t}\, \frac{e^{-ut}}{t^{1/2}} \min\left(1, \frac{|y|}{\sqrt{t}}\right)
   &\leq &\int\limits_{0}^{|y|^2} t^{-1/2}\, \d{t} + |y|\int\limits_{|y|^2}^{\infty} \frac{e^{-ut}}{t}\, \d{t}\\
   &\leq & 2|y| + |y|\left(u^{-1} + 2\ln |y| \right),
\end{eqnarray*}
the desired estimate follows.
\end{proof}

\section{Generalized Krein formula} 
\label{AppendixB}

Let $\HH$ and $\widetilde{\HH}$ be arbitrary (complex) Hilbert spaces, let $H_0\geq 0$ be a self-adjoint operator in $\HH$ and let $A:D(A)\subset\HH\to\widetilde{\HH}$ be densely defined and closed with $D(A)\supset D(H_0)$. Let $J\in \LL(\widetilde{\HH})$ be a self-adjoint isometry and let $B=JA$.

Suppose that $BD(H_0)\subset D(A^{*})$ and that $A^{*}A$ and $A^{*}B$ are $H_0$-bounded with relative bound less than one. Then  
\begin{equation}\label{abstract-ham}
      H = H_0 - A^{*}B
\end{equation}
is self-adjoint on $D(H_0)$ by Kato-Rellich. Operators of the more general form $H = H_0 - gA^{*}B$ with $g\in \R$ can also be written in the form \eqref{abstract-ham} by absorbing $|g|^{1/2}$ in $A$ and $\sgn(g)$ in $J$. For $z\in \rho(H_0)$, let $\phi(z):D(A^{*})\subset \widetilde{\HH}\to\widetilde{\HH}$ be defined by
\begin{align*}
\phi(z) := B(H_0+z)^{-1}A^{*}.
\end{align*}
Note that $D(A^{*})\subset \widetilde{\HH}$ is dense because $A$ is closed.

\begin{theorem}\label{thm:Krein}
Let the above hypotheses be satisfied and let $z\in \rho(H_0)$. Then $\phi(z)$ is a bounded operator. The operator $1-\phi(z)$ is invertible if and only if $z\in \rho(H_0)\cap\rho(H)$, and then
\begin{align}
    (H + z)^{-1} &= R_0(z) + R_0(z)A^{*} (1-\phi(z))^{-1} BR_0(z)\label{res-H}\\
   (1-\phi(z))^{-1} &= 1+ B(H+z)^{-1}A^{*}.\label{res-phi}
\end{align}
\end{theorem}

\noindent
\emph{Remark.} Note that $(1-\phi(z))^{-1}$ leaves $D(A^{*})$ invariant. This follows from \eqref{res-phi} and from the  assumption on $B$.

\begin{proof}
\noindent
\emph{Step 1.}  $A(H_0+c)^{-1/2}$ is bounded for all $c>0$, and  $A(H+c)^{-1/2}$ is bounded for $c>0$ large enough.

By assumption, the operator $H_0-A^*A$ is bounded from below. This implies that 
$$
     \|A\psi\|^2 \leq \|(H_0+c)^{1/2}\psi\|^2 + C \|\psi\|^2
$$
for all $\psi\in D(H_0)$, all $c>0$, and some constant $C$. Since $D(H_0)$ is dense in $D(H_0^{1/2})$ and since $A$ is closed, this bound extends to all $\psi\in D(H_0^{1/2})$ by an approximation argument. In particular, $D(A)\supset D(H_0^{1/2})$. 
The second statement of Step 1 follows from the first and from the fact that $H$ and $H_0$ have equivalent form norms, which implies that $D(H^{1/2}) = D(H_0^{1/2})$.\medskip

\noindent
\emph{Step 2.} If $z\in \rho(H_0)$, then $\phi(z)$ is a bounded operator, and if $z\in \rho(H)$, then 
$$
     \Lambda(z)  := B(H+z)^{-1}A^{*}
$$
is a bounded operator too. This easily follows from Step 1 and from the first resolvent identity.\medskip

\noindent
\emph{Step 3.} If $z\in \rho(H_0)\cap \rho(H)$, then $(1-\phi(z))$ is invertible and $1+\Lambda(z)=(1-\phi(z))^{-1}$.

Both $\phi(z)$ and $\Lambda(z)$ leave $D(A^{*})$ invariant and on this subspace, by straightforward computations using the second resolvent identity, $(1-\phi(z)) (1+\Lambda(z)) = 1 = (1+\Lambda(z)) (1-\phi(z))$. 
\medskip

\noindent
\emph{Step 4.}   If $z\in \rho(H_0)$ and $1-\phi(z)$ is invertible, then $z\in \rho(H)$,  $ (1-\phi(z))^{-1}$ leaves $D(A^{*})$ invariant and \eqref{res-H} holds.

By Step 3, $ (1-\phi(i))^{-1} = 1+\Lambda(i)$, which leaves $D(A^{*})$ invariant. Now suppose that $z\in \rho(H_0)$ and
that $1-\phi(z)$ has a bounded inverse. Then
\begin{align*}
     (1-\phi(z))^{-1}  &=   (1-\phi(i))^{-1} +   (1-\phi(i))^{-1} (\phi(z)-\phi(i))  (1-\phi(z))^{-1}\\
         &=   (1-\phi(i))^{-1} +   (1-\phi(i))^{-1} BR_0(i)(AR_0(\bar z))^{*} (1-\phi(z))^{-1}(i-z).
\end{align*}
Since $\ran BR_0(i)\subset D(A^{*})$, it follows that $(1-\phi(z))^{-1}$ leaves $D(A^{*})$ invariant as well, and that
$$
    R(z) := R_0(z) + R_0(z)A^{*} (1-\phi(z))^{-1} BR_0(z)
$$
is well defined. Now it is a matter of straightforward computations to show that $(H+z)R(z)=1$ on $\HH$ and that $R(z)(H+z)=1$ on $D(H)$. 
\end{proof}


\section{$\Gamma$-Convergence} 
\label{sec:Gamma}

In this section we work in $L^2(\R^N)$ rather than the subspace $\HH$ of symmetric wave functions. In fact, the results of this section are easily generalized to $N$ distinct particles with masses $m_1,\ldots, m_N$ and potentials $V_{ij}\in L^1(\R)$ depending on the pair $i<j$ of particles.

Let $V\in L^1(\R)$, let $g=\lim_{\eps\to 0}g_\eps$, and let $\alpha =g\int V(r)\,dr$. Let $q$ and $q_\eps$ denote the quadratic forms on $H^1(\R^N)$ defined by 
\begin{align*}
   q(\psi) &:= \int |\nabla\psi|^2 + C|\psi|^2\, dx - \alpha \sum_{i<j}\|\gamma_{ij} \psi\|^2\\
   q_\eps(\psi) &:= \int |\nabla\psi|^2 + C|\psi|^2\, dx - g_\eps\sum_{i<j} \int V_{\eps}(x_j-x_i)|\psi|^2\,dx,
\end{align*}
where $C\in \R$ and $\gamma_{ij}:H^1(\R^N)\to L^2(\R^{N-1})$ denotes the trace operator with
$$
     (\gamma_{ij}\psi)(x_1\ldots,x_{j-1},x_{j+1},\ldots x_{N}) = \psi(x_1,\ldots, x_{N})\big|_{x_j=x_i}
$$
for $\psi\in C_0^{\infty}(\R^N)$. It is well known that this operator extends to a bounded operator from $H^1(\R^N)$ to $H^{1/2}(\R^{N-1})$. By Lemma \ref{interaction-bound}, the quadratic forms $q$ and $q_{\eps}$ are bounded below and closed. More precisely, we may choose $C$ so large, that $q\geq 0$ and $q_{\eps}\geq 0$ for all $\eps>0$. 

We are going to prove weak and strong $\Gamma$-convergence $q_{\eps}\to q$ as $\eps\to 0$. To this end, it is convenient to extend all quadratic forms to $L^2(\R^N)$ by setting $q=q_\eps=+\infty$ in $L^2(\R^N)\backslash H^{1}(\R^N)$. The main ingredients of this section are the inequalities
\begin{eqnarray}
         \sup_{r} \int_{\R^{N-1}} |\psi(r,x)|^2\,dx &\leq &\|\partial_r\psi\| \| \psi\| \label{sobo1} \\
        \sup_{r\neq 0} \frac{1}{|r|^{1/2}}  \left|\int_{\R^{N-1}} |\psi(r,x)|^2 -  |\psi(0,x)|^2\,dx \right| & \leq & 2  \|\partial_r\psi\|^{3/2}  \|\psi\|^{1/2} \label{sobo2} 
\end{eqnarray}
for $\psi\in H^{1}(\R^N)$. They are obtained by applying to $\ph(r) = \int |\psi(r,x)|^2\,dx$ the elementary Sobolev inequalities
\begin{align*}
    |\ph(r)| &\leq \frac{1}{2} \int |\ph'(s)|\, ds\\
     |\ph(r)-\ph(0)| &=  \left|\int_0^{r} \ph'(s)\, ds\right| \leq |r|^{1/2} \|\ph'\|.
\end{align*}

\begin{lemma}\label{interaction-bound}
For all $\mu>0$ there exists $C_\mu\in \R$ such that for all $\psi\in H^1(\R^N)$ and $i, j\in \{1,\ldots,N\}$, $i<j$,
\begin{eqnarray}
 \|\gamma_{ij}\psi\| &\leq & \mu\|\nabla\psi\| + C_\mu \|\psi\|,    \label{delta-bound}\\     
 \left|\int V(x_j-x_i)|\psi|^2\, dx\right| &\leq & \|V\|_{L^1}\cdot\|\nabla\psi\| \|\psi\|.      \label{V-bound}
\end{eqnarray}
\end{lemma}

\begin{proof}
Inequality \eqref{delta-bound} follows from $\|\gamma_{ij}\psi\| \leq C\|\psi\|_{H^1}$ by a 
scaling argument. To prove \eqref{V-bound} for $(i,j)=(1,2)$ we set $\tilde\psi(r,R,x) := \psi(R-\tfrac{r}{2},R+\tfrac{r}{2},x)$ and write
\begin{align*}
        \left|\int V(x_1-x_2)|\psi|^2\, dx\right| &\leq \int |V(r)| |\tilde\psi(r,R,x')|^2\, \d{(r,R,x')} \\
        &\leq \|V\|_{L^1} \sup_{r} \int |\tilde\psi(r,R,x')|^2\, \d{(R,x')}
\end{align*}
and we apply  \eqref{sobo1} to the $H^1$-function $\tilde\psi$. Then we use that $\|\tilde\psi\| = \|\psi\|$ and $\|\partial_r\tilde\psi\| \leq \|\nabla\psi\|$. 
\end{proof}

Lemma~\ref{interaction-bound} and $\|V_\eps\|_{L^1} = \|V\|_{L^1}$ imply the corollary:

\begin{kor}\label{uniform-bound}
Under the assumptions of this section, for every $a>0$ there exists $b>0$ such that for all $\eps>0$ and all $\psi\in H^1(\R^N)$,
$$
      (1-a)\|\psi\|^2_{H^1} - b\|\psi\|^2 \leq q_\eps(\psi) \leq (1+a)\|\psi\|^2_{H^1} + b\|\psi\|^2.
$$
\end{kor}

\begin{theorem}\label{Gamma-lim}
If $V\in L^1\cap L^2(\R)$, $\int |V(r)| |r|^{1/2}\, dr<\infty$, and $\alpha =g\int V(r)\,dr$, then $q_\eps \to q$ in the sense of weak and strong $\Gamma$-convergence.
\end{theorem}

\begin{proof}
Due to the fact that all form domains are equal, it suffices to show that, see \cite{ROV2014, DalMaso}, for all $\psi\in H^1(\R^N)$, 
\begin{equation}
\label{Gamma1}
              q(\psi) = \lim_{\eps\to 0} q_{\eps}(\psi)
\end{equation}
and for all $\psi_\eps,\psi\in L^2(\R)$, 
\begin{equation}\label{Gamma2}
     \psi_{\eps} \rightharpoonup \psi\quad \Rightarrow\quad  q(\psi) \leq \liminf_{\eps\to 0} q_{\eps}(\psi_\eps).
\end{equation}    
We begin with the proof of \eqref{Gamma1}. If $\psi\in C_0^{\infty}(\R^N)$, then it is a fairly straightforward application of Lebesgue dominated convergence to show that \eqref{Gamma1} holds. Now let $\psi\in H^1(\R^N)$ and let $(\psi_n)$ be a sequence in $C_0^{\infty}(\R^N)$ with $\psi_n\to \psi$ in $H^1$. Then, on the one hand, 
\begin{equation}\label{Gamma3}
       |q(\psi) - q(\psi_n)| \to 0\qquad (n\to\infty)        
\end{equation}
because $q$ is continuous w.r.t its form norm, which is equivalent to the norm of $H^1(\R^N)$ by Corollary~\ref{uniform-bound}. On the other hand, 
\begin{equation}\label{Gamma4}
     |q_\eps(\psi) - q_\eps(\psi_n)| \to 0\qquad (n\to\infty)
\end{equation}
uniformly in $\eps>0$. This easily follows from Corollary~\ref{uniform-bound}, which means that the interaction is $H^1$-bounded uniformly in $\eps$. Due to \eqref{Gamma3} and \eqref{Gamma4} the validity of \eqref{Gamma1} extends from $C_0^{\infty}(\R^N)$ to $H^1(\R^N)$. 

Now we prove \eqref{Gamma2}. Let $\psi,\psi_\eps\in L^2(\R^N)$ and suppose that $\psi_{\eps} \rightharpoonup \psi$ in $L^2(\R^N)$. To prove  \eqref{Gamma2} we may assume that $\liminf_{\eps\to 0} q_{\eps}(\psi_\eps)<\infty$. We choose a sequence $\eps_n\to 0$ so that $\liminf_{\eps\to 0} q_{\eps}(\psi_\eps) = \lim_{n\to \infty} q_{\eps_n}(\psi_{\eps_n})$. Then, by Corollary~\ref{uniform-bound} it follows that 
$\psi_{\eps_n}$ is bounded in $H^1$ uniformly in $n$. Therefore, after passing to a subsequence, we may assume that $\psi_{\eps_n} \rightharpoonup \tilde\psi$ in $H^1$. Since $\psi_{\eps_n} \rightharpoonup \psi$ in $L^2(\R^N)$ it follows that $\psi=\tilde\psi\in H^1(\R^N)$.

By the weak lower semicontinuity of positive quadratic forms we know that 
$$
        q(\psi) \leq \liminf_{n\to \infty} q(\psi_{\eps_n}).
$$
On the right hand side we may replace $q(\psi_{\eps_n})$ by $q_{\eps_n}(\psi_{\eps_n})$ if we can show that
\begin{equation}\label{qqeps}
     |q_\eps(\psi) - q(\psi)| = o(1)\cdot \|\psi\|^2_{H^1} \qquad (\eps\to 0)
\end{equation}
uniformly $\psi\in H^1(\R^N)$. To prove this, we begin with 
\begin{equation}\label{compare-qs}
  |q_\eps(\psi) - q(\psi)| \leq \sum_{i<j} \left|\int g_\eps V_{\eps}(x_i-x_j)|\psi |^2\,dx - \alpha_\eps\|\gamma_{ij} \psi\|^2\right| + o(1)\cdot \|\psi\|^2_{H^1}, 
\end{equation}
where $\alpha_\eps:= g_\eps\int V(r)\,dr$, and $\alpha_\eps\to \alpha$ has been used.
Applying \eqref{sobo2} to $\tilde\psi(r,R,x) := \psi(R-\tfrac{r}{2},R+\tfrac{r}{2},x)$,
we see that the contribution of $(i,j)=(1,2)$ to \eqref{compare-qs} has the bound
\begin{align*}
  &  \left| g_\eps\int V_{\eps}(r) |\psi(R-\tfrac{r}{2},R+\tfrac{r}{2},x)|^2\, drdRdx - \alpha_\eps\int |\psi(R,R,x)|^2\,dRdx \right| \\
&= \left| g_\eps \int dr V(r) \int \left(  |\psi(R-\tfrac{\eps r}{2},R+\tfrac{\eps r}{2},x)|^2 - |\psi(R,R,x)|^2 \right)dRdx \right| \\
&\leq C |g_\eps| \int |V(r)| |\eps r|^{1/2}\, dr \cdot  \|\psi\|^2_{H^1}.
\end{align*}
Here, we used that $\|\tilde\psi\|_{H^1} \leq C \|\psi\|_{H^1} $. Since all summands of \eqref{compare-qs} can be estimated in this way, \eqref{qqeps} is true and  the proof is complete.
\end{proof}

In view of  Theorem~13.6 in \cite{DalMaso}, Theorem~\ref{Gamma-lim} has the following corollary:

\begin{kor}\label{strongRlim}
If $V\in L^1(\R)$ with $\int |V(r)| |r|^{1/2}\, dr<\infty$ and $\lim_{\eps\to 0}g_\eps=g$, then $H_\eps \to H$ in the strong resolvent sense.
\end{kor}

\bigskip\noindent
\textbf{Acknowledgement.} The second author thanks Jacob Schach M{\o}ller for stimulating discussions and for the hospitality at Aarhus University.  Our work was supported by the \emph{Deutsche Forschungsgemeinschaft (DFG)} through the Research Training Group 1838: \emph{Spectral Theory and Dynamics of Quantum Systems}.


\end{document}